\documentclass[a4paper,10pt]{article}
\usepackage{amsthm,amsfonts,amssymb,euscript}
\usepackage{latexsym, multicol, fancybox}
\usepackage{graphicx}
\usepackage{color}
\usepackage{amsmath, amsthm, amssymb, bm}
\usepackage{epstopdf}
\usepackage{caption}
\usepackage{psfrag}
\usepackage{slashed}

	\def\grad{\mbox{grad}}												

\usepackage{mathrsfs}
 \usepackage{xcolor}
 \usepackage[citebordercolor={green}]{hyperref}
 \usepackage{tikz}
 \usepackage{bbm}
 \usepackage{dsfont}
 \usepackage{bbold}

\numberwithin{equation}{section}

\newtheorem{theorem}{Theorem}[section]
\newtheorem{lemma}[theorem]{Lemma}
\newtheorem{proposition}[theorem]{Proposition}
\newtheorem{corollary}[theorem]{Corollary}
\newtheorem{definition}[theorem]{Definition}
\newtheorem{remark}[theorem]{Remark}

\newcommand{\bea}{\begin{eqnarray}}
\newcommand{\eea}{\end{eqnarray}}
\def\beaa{\begin{eqnarray*}}
\def\eeaa{\end{eqnarray*}}
\def\ba{\begin{array}}
\def\ea{\end{array}}
\def\be#1{\begin{equation} \label{#1}}
\def \eeq{\end{equation}}
\newcommand{\lab}{\lababel}

\def\a{{\alpha}}

\def\b{{\beta}}
\def\be{{\beta}}
\def\ga{\gamma}
\def\Ga{\Gamma}
\def\de{\delta}

\def\la{\lambda}

\def\Si{\Sigma}
\def\om{\omega}

\def\varo{{\varrho}}
\def\th{\theta}

\def\ze{\zeta}

\def\nab{\nabla}

\def\pr{{\partial}}
\def\les{\lesssim}
\def\c{\cdot}

\def\MM{{\mathcal M}}

\def\HH{{\mathcal H}}

\def\OO{{\mathcal O}}
\def\SS{{\mathcal S}}

\def\Div{\mbox{Div}\,}
\def\d{{\mathfrak{d}}}

\def\DD{{\mathcal D}}

\def\RR{{\mathcal R}}

\def\HH{{\mathcal H}}

\def\T{{\bf T}}

\def\f12{{\frac 1 2}}

\def\dual{{\,^*}}
\def\div{\mathrm{div}\,}
\def\curl{\mathrm{curl}\,}

\def\Hb{\,\underline{H}}

\def\hbub{{\underline{H}_{\underline{u}}}}

\def\Xh{\,^{(h)}X}

\def\trch{{\mathrm{tr}}\, \chi}
\def\chih{{\widehat \chi}}
\def\chib{{\underline \chi}}
\def\chibh{{\underline{\chih}}}

\def\etab{{\underline \eta}}
\def\omb{{\underline{\om}}}
\def\bb{{\underline{\b}}}
\def\aa{\protect\underline{\a}}
\def\xib{{\underline \xi}}

\def\Xh{\widehat{X}}

\def\ub{\underline{u}}

\def\varoc{\check{\varrho}}

\def\tr{{\mathrm{tr}}}

\def\trchb{{\tr \,\chib}}

\def\atrch{\, ^{(a)}\trch}
\def\atrchb{\, ^{(a)}\trchb}

\def\wtb{\widecheck{\trchb}}

\def\hot{\widehat{\otimes}}
\def\rhod{\,\dual\hspace{-2pt}\rho}

\def\err{{\mbox{err}}}

\def\f12{\frac 1 2}
\def\lab{\label}

\def\bsplit{\begin{split}}

\def\ab{\underline{\a}}

\def\sk{\mathfrak{s}}

\DeclareFontFamily{U}{mathx}{\hyphenchar\font45}
\DeclareFontShape{U}{mathx}{m}{n}{
      <5> <6> <7> <8> <9> <10>
      <10.95> <12> <14.4> <17.28> <20.74> <24.88>
      mathx10
      }{}
\DeclareSymbolFont{mathx}{U}{mathx}{m}{n}
\DeclareFontSubstitution{U}{mathx}{m}{n}
\DeclareMathAccent{\widecheck}{0}{mathx}{"71}

\def\Gac{\widecheck{\Ga}}

\def\rhoc{\widecheck{\rho}}

\def\trchc{\widecheck{\tr\chi}}
\def\trchbc{\widecheck{\tr\chib}}

\def\pa{\partial}

\def\Kh{\,^{(h)}K}

\DeclareFontFamily{U}{mathx}{\hyphenchar\font45}
\DeclareFontShape{U}{mathx}{m}{n}{
      <5> <6> <7> <8> <9> <10>
      <10.95> <12> <14.4> <17.28> <20.74> <24.88>
      mathx10
      }{}
\DeclareSymbolFont{mathx}{U}{mathx}{m}{n}
\DeclareFontSubstitution{U}{mathx}{m}{n}
\DeclareMathAccent{\widecheck}{0}{mathx}{"71}

\def\SS{\mathcal{S}}

\def\NI{\noindent}
 \def\ub{\underline{u}}
 
 \def\eg{\, ^{(g)}\hspace{-1pt} e}

\def\xibg{\, ^{(g)}\hspace{-1pt} \xib}

\def\ombg{\, ^{(g)}\hspace{-1pt} \omb}
\def\etag{\, ^{(g)}\hspace{-1pt} \eta}
\def\etabg{\, ^{(g)}\hspace{-1pt} \etab}
\def\zetag{\, ^{(g)}\hspace{-1pt} \zeta}
\def\omg{\, ^{(g)}\hspace{-1pt} \om}
\def\trchg{\, ^{(g)}\hspace{-1pt} \trch}

\def\zeg{\, ^{(g)\hspace{-2pt} }\ze}

\def\nabg{\, ^{(g)}\nab}

\def\Gag{\, ^{(g)}\Ga}
\def\Rg{\, ^{(g)}R }

\def\RR{\mathcal{R}}
\def\OO{\mathcal{O}}

\def\SS{\mathcal{S}}

\def\vol{\mathrm{vol}}
\def\o{\mathcal{O}}
\def\T{\, ^{(T)}}
\def\et{\, ^{(T)}\hspace{-1pt} e}
 
\begin{document}

\title{Formation of Trapped Surfaces in Geodesic Foliation}
\author{Xuantao Chen and Sergiu Klainerman}

\date{}
\maketitle
\abstract{We revisit the  classical  results    of the formation of trapped  surfaces   for the Einstein vacuum equation    relying on the geodesic foliation,  rather than the double  null foliation used in all  previous results,  starting with  the  seminal  work of  Christodoulou   \cite{Chr1} and continued  in  \cite{KRodn}, \cite{An}, \cite{AnLuk}, \cite{KLR},  \cite{An1}.   The main advantage of the  method  is that it only  requires   information
 on  the incoming  curvature    along the incoming initial null hypersurface. The result is based on a  version of the non-integrable PT  frame  introduced in  \cite{KS:Kerr} and \cite{GKS}, associated to the geodesic foliation.
}

%%%%%%%%%

\section{Introduction}

All known results  on the formation of trapped  surfaces   for the Einstein vacuum equation    starting with  the  seminal  work of  Christodoulou   \cite{Chr1} and continued  in  \cite{KRodn},  \cite{AnLuk}, \cite{KLR},  \cite{An1}, \cite{An}, make use of an adapted double null foliation. The goal  of this paper  is  to show that   similar results  can be derived using instead  a simple   geodesic foliation  and  an associated, non-integrable,  PT frame  first introduced in \cite{KS:Kerr}, \cite{GKS}.  
 The main advantage of the  method  is that it only  requires   information
 on  the incoming  curvature    along the incoming initial null hypersurface. The result is based on a  version of the non-integrable PT  foliation introduced in  \cite{KS:Kerr} and \cite{GKS} and uses  the $(a, \de) $  version of the short pulse method introduced in  \cite{AnLuk}.

\subsection{Set-Up}

Consider a spacetime  $\MM=\MM(\delta, a;\tau^*)$ with past null boundaries $\Hb_0\cup H_{-1}$  and future boundaries  $ \Hb_\de\cup\Si_{\tau^*}$, where  $\Hb_\de$ is null incoming and  $\Si_{\tau}$ is a spacelike level hypersurface of a time function $\tau$ to be specified (See Figure 1).  Here $\delta$ is a small constant and, following \cite{AnLuk},  we introduce another  large constant $a$ which  satisfies $a\delta\ll 1$. The spacetime $\MM$ is  foliated  by the level surfaces of an ingoing  optical function $\ub$ such that $\ub=0$ on $ \Hb_0$ and    $\ub=\de$ on $\Hb_\de$.

\begin{figure}
\begin{tikzpicture}[scale=0.8]
\draw [white](3,-1)-- node[midway, sloped, below,black]{$H_{-1}$}(4,0);
\draw [white](0.1,1.9)-- node[midway,above,black]{$\Sigma_{\tau^*}$}(2,2);
\draw [white](2,2)--node [midway,above,black] {$\Hb_\delta$}(4,0);
\draw [white](0,2)--node [midway, below,black] {$\Hb_0$}(3,-1);
\draw [dashed] (0, 2)--(0, -4);
\draw [dashed] (0, -4)--(4,0);%--(0,4);
\draw [dashed] (0,2)--(0.1,1.9);
\draw [dashed] (0,-4)--(4,0);
\draw [very thick] (0.1,1.9)--(3,-1)--(4,0)--(2,2)--(0.1,1.9);
\fill[yellow] (0.1,1.9)--(3,-1)--(4,0)--(2,2)--(0.1,1.9);
\end{tikzpicture}
\centering
    \caption{The spacetime $\MM$}
    \label{fig:enter-label}
\end{figure}

\vspace{2ex}

\NI {\bf Geodesic foliation on $H_{-1}$.}  The  restriction of   $\ub$  to   $H_{-1}$     coincides with the affine parameter of  a null  geodesic  generator  of  $H_{-1}$,  denoted by $e_4$, normalized on the sphere $S_{-1,0}:=H_{-1}\cap \Hb_0$. We let $\ub=0$ on $S_{-1,0}$. This gives a geodesic foliation on $H_{-1}$, and the level surfaces of $\ub$ are $2$-spheres.   We then have:
\begin{equation*}
    \om=\xi=0,\quad \etab=-\zeta.
\end{equation*}
We can also derive the bounds of other   Ricci coefficients,   see Proposition \ref{Prop:InitailData-H_{-1}}.

\vspace{2ex}

 \NI {\bf Geodesic foliation on $\MM$}.    Using the  incoming optical function  $\ub$  we define\footnote{Recall that, given a function $f$, $(\grad f)^\mu:=g^{\mu\nu} \pr_\nu f$.}   $\eg_3:=-2\grad \ub$,  such  that $\eg_3$ is geodesic. We also define $s$  to  be   the affine parameter of $e_3$, i.e.    $\eg_3(s)={1}$  with  $s=-1$ on $H_{-1}$.  We then  define $\eg_4$ to be the null companion of $\eg_3$ orthogonal to the sphere $S_{\ub,s}$, defined as the intersection of level hypersurfaces of $\ub$ as $s$, and denote  by      $\SS$ the horizontal structure      perpendicular  on $\eg_3, \eg_4$, tangent to 
 the spheres  $S_{\ub,s}$.    We also denote   $\nabg, \nabg_3, \nabg_4$   the corresponding  horizontal derivative  operators (see Section \ref{section:horizontalStr}) and by    
  $(\eg_a)_{a=1,2} $ an arbitrary  orthonormal frame  of $\SS$.
  Note that we have
 \bea
 \eg_4(\ub)=g(\grad\ub,\eg_4)=-\frac 1 2 g\big(\eg_3, \eg_4\big)=1.
 \eea
 In particular, restricted  to  $H_{-1}$,  $\eg_4$ coincides with the $e_4$ defined  above on $H_{-1}$. 
 \begin{remark}
 The geodesic foliation and   its associated geodesic  horizontal    null   structure  defined above are a simple  example of a   principal geodesic (PG)  structure,  as introduced in  \cite{KS}.   We   will  thus refer to      $\eg_3, \eg_4, \eg_1,\eg_2$  as a PG frame.
 \end{remark}
 We  denote by $\Gag$  the corresponding  Ricci coefficients  and by  $\Rg$   the   null  curvature  components   with respect to the geodesic frame.  
    Thus (see\footnote{The relations \eqref{eq:intro-geodesic} follow easily  from $D_{\eg_3}\eg_3=0$ and  by applying  the  commutation relations  (see  formula 2.2.3  in \cite{GKS} for an easy derivation)  below to the  functions $\ub, s$
      \begin{equation*}
    \begin{split}
    [\eg_3,\eg_a]&=\, ^{(g)}\xib_a\eg_4+(\etag-\zetag)_a\eg_3-\, ^{(g)}\chib_{ab}\eg_b,\\
    [\eg_4,\eg_a]&=\, ^{(g)}\xi_a \eg_3+(^{(g)}\etab+\zetag)_a \eg_4-\, ^{(g)}\chi_{ab}\eg_b,\\
    [\eg_4,\eg_3]&=2(\, ^{(g)}\etab-\etag)_a \eg_a+2\omg \eg_3-2\, ^{(g)}\omb \eg_4.
    \end{split}
\end{equation*} } e.g. \cite{KS:Kerr}),
\bea
\lab{eq:intro-geodesic}
    \, ^{(g)}\omb=\, ^{(g)}\xib=0, \quad \etag=\zetag=-\, ^{(g)}\etab, \quad  \eg_3\big(\eg_4(s)\big)=-2\omg.
\eea

We associate  to the geodesic frame a system of angular  coordinates $\th^a$, $a=1,2$  as follows:
\begin{itemize}
\item On $H_{-1}$   we set $\eg_4(\th^a)=0 $  with $\th^a$   specified on $S_{0,-1}:=\Hb_0\cap H_{-1}$;
\item Using the   values of $\th^a$   on $H_{-1}$   we extend   them     to        $\MM$   by  $\eg_3(\th^a)=0 $.
\end{itemize}
We define  the  time\footnote{We show later in Section \ref{subsect:Integrating region}  that $\tau$  is indeed a time function.}  function $\tau:=\frac 1{10} a\ub+s$.

\vspace{2ex}

  \NI {\bf PT frame}. In the geodesic frame, each non-vanishing Ricci coefficient satisfies a transport equation along the integral curve of $\eg_3$. Some of  these equations, however,    contain transversal derivatives,  leading to a loss of derivatives. We deal with the issue by considering another     frame $\{\, ^{(T)}e_3,\T e_4,\T e_a\}$  which 
   verifies\footnote{{The existence of such a gauge    can be easily  justified in view of  the   transformation formula \eqref{eq:transformation-eta}.}}
  \begin{equation}
\lab{eq:eta=0}
\T e_3=\eg_3, \qquad     \T \eta=0.
\end{equation}
The remarkable feature of the frame, called PT frame in \cite{KS:Kerr},   is that    the loss of derivative issue   disappears once we set up this  gauge condition, see Proposition \ref{Prop:transport.e_3}. This positive  feature   is however  compensated  by a negative one, that is the fact that the horizontal structure 
 associated to the null pair $(\et_3, \et_4)$ is not integrable,  see Section \ref{section:horizontalStr} and the more detailed discussion in  Chapter  2 of  \cite{GKS}. This  problem  can however be resolved by  relying on both frames, the  non-integrable  PT  frame    to deal with    the $e_3$-transport equations and  the integrable PG frame for   dealing with elliptic and Sobolev type  estimates.
\begin{remark} In what follows, as  there is no danger of confusion,  we drop the prefix  $^{(T)}$  for the PT foliation. We thus  denote  $\et_3 =e_3, \et_4=e_4$, by $\HH$ the  horizontal structure perpendicular to $e_3, e_4$  and by $\nab, \nab_3, \nab_4$   the corresponding derivative operators.
 We  denote by $\Ga=\big\{\trch, \trchb, \atrch, \atrchb, \chih, \chibh$, $ \eta, \etab, \ze, \om, \omb, \xi, \xib\big\}$    the   set   of all PT-Ricci coefficients. 
\end{remark}

 \begin{remark}   The PT frame   we  work with
 coincides with the  geodesic frame on $H_{-1}$ and verifies,   see Definition \ref{def:PT-frame} and 
 Proposition \ref{Prop:PTframe-properties}, 
\bea
\lab{eq:PTframe-property1-intro}
 \omb=\, \xib=0,\qquad  \etab=- \zeta, \qquad  \atrchb=0.
\eea 
\end{remark}
\NI We  introduce   the  renormalized   quantity\footnote{In contrast, since $\trch$ presents a worse behavior similar to $1/|s|$, one does not need to renormalize it by subtracting its Minkowskian value.}
\bea
\trchbc:=\trchb+\frac{2}{|s|}
\eea
and denote by $\Gac$  the  set  of non-vanishing Ricci coefficients  
\beaa
\Gac=\Big\{ \trch, \trchbc, \atrch,  \chih, \chibh,  \ze, \om,  \xi \Big\}. 
\eeaa

\subsection{ Initial conditions}

\NI {\bf  Initial Data on  $\Hb_0$.}  Following   the  results  of   \cite{Chr1},  \cite{KRodn},  \cite{AnLuk}, \cite{KLR},    we start by assuming that the incoming data  on $\Hb_0$  is  Minkowskian\footnote{One can also study other type of incoming data.  In \cite{Li} and \cite{An}, the incoming data  corresponds  to  Christodoulou's naked singularity  solution in \cite{Chr-naked}.}.  We note 
 however  that  {we can significantly relax this assumption by only requiring information on the incoming curvature. Indeed, unlike the case of   the double null foliation used   in these above mentioned  works, all Ricci coefficients  in the    PT frame   can be determined by integration   along the $e_3$ direction. The  incoming data  on $\Hb_0$ is only used    in the derivation of  the curvature components  by energy estimates.

\medskip 

\NI {\bf  Initial Data on  $H_{-1} $.}     Our data  on $H_{-1}$  verifies  the An-Luk   \cite{AnLuk}  short pulse  assumption
\begin{equation}\lab{eq:chih_0-upperbound}
    \sum_{i\leq N_0,j\leq 1}||(\delta\nab_4)^j\nab^i\chih_0||_{L^\infty_{\ub}L^2(S_{-1,\ub})}\leq C_0\, a^\frac 12, \quad N_0\geq 9,
\end{equation}
as well as  
\begin{equation}\lab{eq:chih_0-lowerbound}
    \inf_{\theta}\int_0^\delta |\chih_0(\ub,\th)|^2 d\ub\geq \delta a.
\end{equation}
\begin{remark}
    Using the argument in \cite{KLR} (see also \cite{AnHan}), one can relax \eqref{eq:chih_0-lowerbound} by replacing the inf over $\theta$ by sup. See Remark \ref{rmk:aniso-Ch8}.
\end{remark}

\begin{remark}  
\label{remark:PG=PT-onH_0}
 Note that the $S$-foliation on    $H_{-1}$ is that induced by the geodesic foliation and that     both    the PT  and geodesic frames  discussed  above coincide with double null  frame  on $H_{-1}$ used in \cite{AnLuk} and all the  other above mentioned works. We  point out that in \cite{AnLuk}  the assumption is weaker as there is no requirement on the $\nab_4$ derivative of $\chih_0$ in \eqref{eq:chih_0-upperbound}.  This is achieved by a renormalization of curvature components such that the contribution from $\nab_4\chih_0$ completely decouples from the system.  This can, in principle, be also achieved in our framework but we do not pursue this here.
\end{remark}

\subsection{Main result}
Here is a short version of our main result.\begin{theorem}
\label{theorem:main}
    Consider the characteristic initial value problem described above.  If \eqref{eq:chih_0-upperbound} holds, then the spacetime can be extended to $\MM(\delta,a;-\frac 18 a\delta)$, together with its incoming geodesic foliation. Moreover, if \eqref{eq:chih_0-lowerbound} also holds, then $S_{\delta,-\frac 14 a\delta}$ is a trapped surface\footnote{Note that $\tau(\delta,-\frac 14 a\delta)=\frac{1}{10} a\delta-\frac 14 a\delta<-\frac 18 a\delta$, so $S_{\delta,-\frac 14 a\delta}$ indeed lies in  $\MM(\delta,a;-\frac 18 a\delta)$.}.
\end{theorem}
We later provide  (see Section \ref{section:Main Thm-precise}) a  more precise version of  Theorem \ref{theorem:main}  which also extends to  more general  incoming initial data  on $\Hb_0$.

\medskip

\NI {\bf  Previous results.} 
Christodoulou’s  pioneering 
work \cite{Chr1} is the first result on the formation of trapped surfaces in the Einstein-vacuum spacetime. Klainerman-Rodnianski \cite{KRodn} then adopted a systematical approach by scale invariant estimates to simplify the proof of \cite{Chr1}. This idea was then further generalized by An \cite{An12}. Li-Yu \cite{LiYu} showed that there exists  Cauchy initial data corresponding to Christodoulou's spacetime. Later, Klainerman-Luk-Rodnianski \cite{KLR} significantly relaxed the lower bound in \eqref{eq:chih_0-lowerbound}  by developing a fully anisotropic mechanism for the formation of trapped surfaces.

The first scale-critical result was established by An-Luk \cite{AnLuk}, which led to the further study of the apparent horizon \cite{An17}, \cite{AnHan}. Later An \cite{An1} gave a simplified proof of the scale-critical result in the far-field regime by designing a scale-invariant norm based on the signature and decay rates. Our  work provides proof of a similar result in the finite region using the incoming geodesic foliation instead of the double null foliation.

We also refer the readers to the results generalized to Einstein equation coupled with matter fields \cite{Yu11}, \cite{AnAthanasiou}, \cite{AMY}, \cite{ZHK}.

\subsection{Main  features   in the proof of Theorem \ref{theorem:main} }

\NI  {\bf 1.}   As mentioned earlier  we  make  essential use  of  PT frame \eqref{eq:eta=0} in order to avoid  the loss of derivatives intrinsic to the geodesic foliation. Note that  the horizontal  structure  spanned by $( e_1,  e_2)$ is non-integrable\footnote{It is however integrable  in $e_3$, i.e. $\atrchb=0$.  This is due to the fact that   the corresponding   horizontal structure   is tangent to the $\ub$ hypersurfaces, see  Proposition \ref{Prop:PTframe-properties}.  }  with respect to $e_4$,
     i.e. $\atrch\neq 0$.
    %\begin{equation*}
    %    g([e_a,e_b],e_3)=-\chib_{ab}+\chib_{ab}=0,\quad g([e_a,e_b],e_4)=-\chi_{ab}+\chi_{ba}\neq 0,
    %\end{equation*}
The use of  non-integrable structures   was pioneered   in the proof of Kerr stability \cite{KS:Kerr}, \cite{GKS}.  To compensate  for    the lack of integrability  of the  main horizontal structures  used  in these works  one  needs to  consider  associated integrable structures  for which one can derive  elliptic (Hodge  estimates)   and Sobolev inequalities. In our case  this role is played  by  the geodesic frame $\{\eg\}$.
 
 \medskip 
 
 \NI {\bf 2.} 
  The typical  transport equation      verified   by all  Ricci coefficients in the PT frame,   is of the form
    \begin{equation}
    \lab{eq:model-transport}
        \nab_3\psi+\lambda \trchb\psi=F.
    \end{equation}
    The  main contribution of  $\trchb$  is  $-2/|s|$, so neglecting  $F$ (which we expect to  control by a bootstrap assumption), we  infer that the quantity $|s|^{2\lambda}\psi$ is  conserved. It helps to   divide   all    Ricci coefficients  $\Gac$     as follows:
    \begin{enumerate}
        \item[i.] Those  that  are of size $1$ on $H_{-1}$, and satisfy the transport equation \eqref{eq:model-transport}  with $\lambda=\frac 12$.   They behave  like $1/|s|$ on $\MM$.  We    denote  these    by $\Gac_b$ (stands for ``bad");
        \item[ii.] Those  that are of size  $\delta a^\frac 12$ on $H_{-1}$, and satisfy the transport equation with $\lambda=1$. They behave  like  $\delta a^\frac 12/|s|^2$ on $\MM$ and are denoted by $\Gac_g$ (stands for ``good").
        \item[iii.]  The  outgoing shear $\chih$  for which we   have  only  the bounds  $\chih\sim a^\frac 12/|s|$. In addition, in contrast   to the  case   of   the double null foliation (used  in  \cite{KRodn},  \cite{AnLuk}, \cite{KLR},  \cite{An1}, \cite{Li}, \cite{An}),   we also have present   the  signature\footnote{The signature of a quantity is basically the number of $e_4$ minus the number of $e_3$ in its expression.}  $+2$ quantity $\xi$   which behaves  similarly to $\chih$.
          Though $\xi$, like $\chih$,  is  a large quantity, due to   signature  consideration,  it   gets   paired with   better behaved quantities  in nonlinear terms.

    \end{enumerate}
   \medskip 

\NI   {\bf 3.}  The right hand side  of \eqref{eq:model-transport}  denoted by  $F$  contains  linear   curvature components and nonlinear  terms   relative to  the    Ricci coefficients $(\Gac_g, \Gac_b)$. As usual,  the   curvature components are  controlled   by energy type estimates   using  the  null Bianchi  equations. The  nonlinear quadratic  terms  for (i)   are of the form $\Gac_g\cdot\Gac_b$. Those   for (ii) are  of the  form   $\Gac_b\cdot \Gac_b$.   Discounting the anomalous behavior discussed  below,
both of these will result in the gain of  at least an extra $a^{-\frac 12}$ factor in their   estimates\footnote{The absence of worse nonlinear terms is  related to the ``signature conservation" pointed out in \cite{CK}.}.
In our case the terms in  $\Gac_b$  have  signature $+1$ and those  in  $\Gac_g$  have  signature $0$ or $-1$.  By simple signature  considerations  it is easy to see that 
when   \eqref{eq:model-transport} is  applied  to  $\Gac_g$  we cannot have  $\Gac_g\cdot\Gac_g$ terms in $F$.  Similarly, when  \eqref{eq:model-transport}  is applied to   $\Gac_b$ one cannot have  terms of the type  $\Gac_g\cdot\Gac_b$. 
The absence of a worse term is crucial to  close our estimates.  For example, suppose $\psi\in \Gac_b$ and we have the equation
\begin{equation*}
    \nab_3 \psi+\frac 12 \trchb \psi=\Gac_b\cdot \Gac_b+\cdots
\end{equation*}
Since $\Gac_b\sim 1/|s|$, we would end up integrating $1/|s|$ which gives an additional logarithmic growth in $|s|$.

\medskip

 \NI {\bf 4.}  Anomalies: The key quantity $\chih$ is  large even compared with $\Gac_b$, with an extra $a^\frac 12$ factor. This is  a crucial feature   for the mechanism of the  formation of trapped surfaces.
  Its presence makes some nonlinear terms         become borderline. To overcome this  difficulty  one needs to make use of the  triangular structure of  the   main  $e_3$ transport      equations, that is   to follow  a specified,  correct,  order   in  doing the estimates. 

\medskip 

\NI  {\bf 5.}  As already  mentioned  we  need to work with  both  the geodesic and  PT frames. The passage  from the $\eg$-frame to the $e$ frame 
 is made using the frame transformation formulas (see \eqref{eq:PT-frame}) 
 \beaa
    e_4=\eg_4- f^a \,  \eg_a+\frac 14|f|^2 \eg_3,\quad e_a=\eg_a-\frac 12 \,  f _a \eg_3,  \quad e_3=\eg_3
\eeaa
where $f$ verifies\footnote{This follows  by using the condition $\etag=\zetag$ and  the transformation formulas of  Lemma \ref{Lemma:Generalframetransf}. }

\begin{equation}\label{eq:eq-of-f-intro}
    \nab_3 f+\frac 12\trchb\, f=2\zeta-\chibh\cdot f.
\end{equation}
   The  Ricci  coefficients  in the $e$ and $\eg$   frames    are related by Lemma     \ref{Lemma:Generalframetransf}.    The  $\eg$ frame  is used to derive   Hodge  elliptic   estimates  and  Sobolev inequalities. More precisely, whenever  we need to  make use of these, we pass from the PT frame $e$  to the  $\eg$ frame and then transform the result back to the  PT-frame.

\medskip

 \NI {\bf 6.}   The ansatz from the {bootstrap assumption (see \eqref{eq:bootstrap})} $\zeta\in\Gac_g\sim \delta a^\frac 12 |s|^{-2}$   (which is true in the double null frame\footnote{ and in fact, as one can later verify, also in the integrable geodesic frame}) leads to a logarithmic loss in the $|s|$-weighted estimate  when we integrate the equation \eqref{eq:eq-of-f-intro}. To avoid this  problem   we show that in fact,  in the PT frame,  $\zeta $ satisfies a slightly improved estimate of the form  $\ze\sim \delta a^\frac 12 |s|^{-1}+\delta^\frac 32 a |s|^{-\frac 52}$ that circumvents the problem.

\medskip 

\NI  {\bf 7.}  Apart from the transport equations of type  \eqref{eq:model-transport}   verified  by the  Ricci coefficients, we also need to  control  the curvature  components by using  energy type  estimates\footnote{This is in fact   the only place  where   we  need to take  into account  the  incoming  data on $\Hb_0$. }.   This is a  standard procedure,  see for example Section 8.7 in   \cite{KS} or  Chapter 16 in \cite{GKS}.
   A typical pair of null Bianchi equations  can, in our case, be  written  in  the form
 \begin{equation}
 \lab{eq:Bianchi-pairs}
    \begin{split}
        \nab_3 \psi_1+\lambda\trchb \psi_1&=\DD^*\psi_2+F_1,\\
        \nab_4 \psi_2&=\DD\psi_1+ F_2,
    \end{split}
\end{equation}
where $\DD$, $\DD^*$ represent horizontal Hodge operators (defined in Section \ref{subsect:general-Bianchi-pair}) that are   formal adjoint of each other. The  corresponding energy estimates  for  $(\psi_1, \psi_2)$   is derived by integrating  the   divergence identity 
\begin{equation*}
    \Div(|s|^{2(2\lambda-1)}|\psi_1|^2 e_3)+\Div(|s|^{2(2\lambda-1)}|\psi_2|^2 e_4)=\cdots
\end{equation*}
on the  causal region  enclosed by the boundaries  $\Hb_0$, $H_{-1}$, $\Hb_\delta$, $\Sigma_\tau$.

\medskip 

\NI  {\bf 8.}    To estimate  higher derivatives  we need to commute both the transport equations of type \eqref{eq:model-transport} and the null Bianchi pairs   with $\nab$,  more precisely\footnote{ This latter   can be thought of as the ``rotation operator" which is commonly used in the analysis of wave equations.}   with  $|s|\nab$.  A  small  difficulty  appear  when we commute  the  second equation of the Bianchi pair  \eqref{eq:Bianchi-pairs},   applied to  $\psi_1=\bb, \psi_2=\aa$,   with $\nab$  due to  the   commutator   $   [\nab_4,\nab]\aa=\xi\nab_3 \psi+\cdots$  which contains the term  $ \nab_3\aa$ for which   we do not  have an equation.

 It turns out that, with very little additional work, we can also commute  equation  \eqref{eq:model-transport} with 
$|s|\nab_3$ just as  with   $|s|\nab$. As a result, $|s|\nab\psi$ and $|s|\nab_3\psi$ both behave similarly with  $\psi$. We note however that the signature of $\nab$ and $\nab_3$ are different, and this is the reason why we do not pursue the strict hierarchy according to the signatures, as in \cite{KRodn}, \cite{An1},  but only distinguish $\nab_4$ with $\d=(\nab,\nab_3)$, and, by a similar spirit, distinguish $\Gac_b$, which is of signature $+1$, with $\Gac_g$, of signature $0$ or $-1$.

We also note that the analogous problem  of the commutator between $\nab_3$ and $\nab$ is not present in view of the  fact  $\xib=0$  in  our PT gauge.  We rely very little   on 
the  $\nab_4$ transport  equations for the Ricci coefficients - they  are in fact only needed on $H_{-1}$.

\section{Preliminaries}
\subsection{Horizontal structures}
\lab{section:horizontalStr}
We review below   some  basic facts about non-integrable horizontal structures discussed in  Chapter 2 of \cite{GKS}.

Given a pair of null vectors $\{e_3,e_4\}$   satisfying $g(e_3,e_4)=-2$, we consider the  horizontal structure  associated to it given by    the distribution  $\HH=\{e_3,e_4\}^\perp$. With a choice of an orthonormal basis $\{e_1,e_2\}$ of this horizontal structure, we obtain a null frame $\{e_3,e_4,e_a\}$ ($a=1,2$). When the horizontal structure is integrable, i.e. the  distribution $\HH$  is involutive,   we also say that the null frame is integrable (which is not the case for the principal null pair in Kerr spacetime).

The Ricci coefficients and curvature components are defined\footnote{Here $\dual R$ is defined by $\dual R_{\a\b\mu\nu}=\frac 12 \in_{\mu\nu}^{\quad \rho\sigma} W_{\a\b\rho\sigma}$.} by
\beaa
    \chi_{ab}&=&g(D_a e_4,e_b),\quad \chib_{ab}=g(D_a e_3,e_b),\quad \xi_a=\frac 12 g(D_4 e_4,e_a),\quad\, \xib_a=\frac 12 g(D_3 e_3,e_a),\\
    \om&=&\frac 14 g(D_4 e_4,e_3),\quad \omb=\frac 14 g(D_3 e_3,e_4),\quad \eta_a=\frac 12 g(D_3 e_4,e_a),\quad \etab_a=\frac 12 g(D_4 e_3,e_a), \\
    \zeta_a&=&\frac 12 g(D_a e_4,e_3),
\eeaa
\beaa
    \a_{ab}&=&R_{a4b4},\quad \b_a=\frac 12 R_{a434},\quad \rho=\frac 14 R_{3434},\quad \dual\rho=\frac 14\dual R_{3434},\quad \bb_a=\frac 12 R_{a334},\\
     \aa_{ab}&=&R_{a3b3}.
\eeaa

For a vector field $X$, we define its projection onto the horizontal structure $\HH$ by
\begin{equation*}
    \Xh:=X+\frac 12 g(X,e_3)e_4+\frac 12 g(X,e_4)e_3.
\end{equation*}
This also defines the projection operator $\Pi$. A $k$-covariant tensor field $U$ is called horizontal, if 
\begin{equation*}
    U(X_1,\cdots, X_k)=U(\Xh_1,\cdots,\Xh_k).
\end{equation*}
The horizontal covariant derivative operator $\nab$ is defined by
\begin{equation*}
    \nab_X Y:=\, ^{(h)}(D_X Y)=D_X Y-\frac 12 \chib(X,Y)e_4-\frac 12 \chi(X,Y)e_3
\end{equation*}
using the definition of $\chi,\chib$. Similarly, one can define $\nab_3 X$ and $\nab_4 X$ as the projections of $D_3 X$ and $D_4 X$. Then the horizontal covariant derivative can be generalized for tensors in the standard way
\begin{equation*}
    \nab_Z U(X_1,\cdots,X_k)=Z (U(X_1,\cdots X_k))-U(\nab_Z X_1,\cdots, X_k) -\cdots - U(X_1, \cdots, \nab_Z X_k),
\end{equation*}
and similarly for $\nab_3 U$ and $\nab_4 U$.

%We have the following Leibniz rule
%\begin{equation*}
%    Z(|\psi|^2)=2D_Z \psi\cdot \psi=2\nab_Z \psi\cdot \psi.
%\end{equation*}

In the non-integrable case, the null second fundamental forms are  decomposed as
\begin{equation*}
    \chi_{ab}=\chih_{ab}+\frac 12\delta_{ab} \trch+\frac 12\in_{ab}\atrch,
\end{equation*}
\begin{equation*}
    \chib_{ab}=\chibh_{ab}+\frac 12\delta_{ab} \trchb+\frac 12\in_{ab}\atrchb.
\end{equation*}
where the trace and anti-trace are defined by
\begin{equation*}
    \trch=\delta^{ab} \chi_{ab},\quad \trchb=\delta^{ab}\chib_{ab},\quad \atrch:=\in^{ab}\chi_{ab},\quad \atrchb:=\in^{ab}\chib_{ab},
\end{equation*}
and the horizontal volume form $\in_{ab}$ is defined by
\begin{equation*}
    \in (X, Y) :=\frac 12\in (X,Y,e_3,e_4).
\end{equation*}
The horizontal structure $\HH$ is integrable if and only if   $\atrch=\atrchb=0$; see  Chapter 2  in \cite{GKS}.

The left dual of a horizontal $1$-form  $\psi$ and a horizontal covariant $2$-tensor $U$  are defined by
\begin{equation*}
    \dual \psi_a:=\in_{ab}\, \psi_b,\quad (\dual U)_{ab}:=\in_{ac}\, U_{cb}.
\end{equation*}
For two horizontal $1$-forms $\psi$, $\phi$, we also define
\begin{equation*}
    \psi\cdot \phi:=\delta^{ab}\psi_a\phi_b,\quad \psi\wedge\phi:=\in^{ab} \psi_a \phi_b,\quad (\psi\hot\phi)_{ab}=\psi_a\phi_b+\psi_b\phi_a-\delta_{ab}\psi\cdot\phi.
\end{equation*}
In particular $|\psi|:=(\psi\cdot\psi)^\frac 12$ with the straightforward generalization to general horizontal covariant tensors. This will be used to define $L^p$-type norms of $\psi$. Similarly we define the derivative operators
\begin{equation*}
    \div \psi:=\delta^{ab}\, \nab_a \psi_b,\quad \curl \psi:=\, \in^{ab} \nab_a \psi_b,\quad (\nab\hot \psi)_{ab}:=\nab_a \psi_b+\nab_b \psi_a-\delta_{ab}\, \div\psi.
\end{equation*}

%%%%%%%%%%%%%%%
\subsection{Frame transformations}
To pass from the geodesic frame $\eg$ to the  PT frame   we need to appeal to the transformation formulas  for the corresponding Ricci coefficients   given in Section 2.2  of \cite{KS:Kerr}.
The general  formula of a transformation between two  null frames  $e$ and $e'$  was given  in Lemma 2.2.1 of that section. In our  context we only need transformations that preserve $e_3$:

\begin{lemma}
\label{Lemma:Generalframetransf}
 A general null transformation between two  null frames $(e_3, e_4, e_1, e_2)$ and $(e_3', e_4', e_1', e_2')$  which preserves $e_3$   has the form
 \bea
 \lab{General-frametransformation}
  e_3'= e_3, \qquad    e_a'= e_a +\frac 1 2 f_a   e_3,\qquad  e_4'= e_4 + f^b e_b+\frac 1 4 |f|^2 e_3 
 \eea
 The inverse  transformation which takes the $e$ frame to the $e'$ frame  is given by replacing $f$ with $-f$, i.e.
 \beaa
  e_3= e'_3, \qquad    e_a= e'_a -\frac 1 2 f_a   e'_3,\qquad  e_4= e'_4 - f^b e'_b+\frac 1 4 |f|^2 e'_3.
 \eeaa
  \end{lemma}

\begin{lemma}
\lab{Proposition:transformationRicci}
Under a null frame transformation \eqref{General-frametransformation} the Ricci coefficients  transform as follows:
\begin{itemize}
\item The transformation formula for $\xi$ is given by 
\bea
\bsplit
\xi' &= \xi +\frac{1}{2}\nab_4'f+\frac{1}{4}(\trch f -\atrch\dual f)+\om f+\frac{1}{2}f\c\chih+\frac{1}{4}|f|^2\eta \\
&+\frac{1}{2}(f\c \ze)\,f +\frac{1}{4}|f|^2\etab -\frac 14|f|^2 \om f+\frac{1}{16}|f|^2 f\cdot\chib+\frac 1{16}|f|^4\xib.
       \end{split}
\eea

\item The transformation formula for $\xib$ is given by 
\bea
\bsplit
\xib' &= \xib,
       \end{split}
\eea

\item The transformation formulas for $\chi $ are  given by 
\bea
\bsplit
\chi_{ab}'&=\chi_{ab}+f_a\eta_b+\nab_{e_a'}f_b+\frac 14 |f|^2 \chib_{ab}+\frac 14|f|^2 f_a\xib_b+f_b\zeta_a\\
&-f_a f_b\omb-\frac 12 f_b f_c\chib_{ac}-f_af_bf_c \xib_c.
\end{split}
\eea

\item The transformation formulas for $\chib $ are  given by 
\bea
\chib_{ab}'=\chib_{ab}+f_a\xib_b.
\eea

\item  The transformation formula for $\ze$ is given by 
\bea
\ze' = \ze   -\frac{1}{4}\trchb f -\frac{1}{4}\atrchb \dual f  -\omb f -\frac{1}{2}\chibh\c f -\frac 12 (f\cdot\xib)f.
\eea

\item   The transformation formula for $\eta$ is given by 
\bea\label{eq:transformation-eta}
\eta' = \eta +\frac{1}{2}\nab_3 f -\omb\, f +\frac 14 |f|^2\xib-\frac 12(f\cdot \xib)f.
\eea
\item   The transformation formula for $\etab$ is given by 
\bea
\etab' = \etab +\frac{1}{4}\trchb f - \frac{1}{4}\atrchb\dual f + \frac{1}{2}f\c\chibh+\frac 14 |f|^2 \xib.
\eea

\item   The transformation formula for $\om$ is given by
\bea
\bsplit
\om' &=  \om +\frac{1}{2}f^a (\ze-\etab)_a -\frac{1}{4}|f|^2\omb - \frac{1}{4}f^a f^b \chib_{ab} -\frac 18 |f|^2 f^a\xib_a.
\end{split}
\eea

\item   The transformation formula for $\omb$ is given by
\bea
\omb' = \omb +\frac{1}{2}f\c\xib.
\eea
\end{itemize}
\end{lemma} 
The proof follows from a direct calculation. See \cite{KS-GCM1} for detailed derivations in full generality. Note that, unlike the  version in \cite{KS:Kerr},  we keep  track of all  error terms\footnote{ This is needed as  our  situation here is non-perturbative.}.

%%%%%%%%%%

\subsection{Passage from the  PG  to the PT frame}
Consider the transformation formula  from the  PG  frame $\eg$  to  a new frame $e $         
for which  $\eta=0$. In view of Lemma \ref{Proposition:transformationRicci}, with  $f$ replaced\footnote{Thus $f$ corresponds to the inverse  transformation formula  from the $e$-frame to the $\eg$-frame.}   by $-f$,  we must have
\beaa
\eta = \etag -\frac{1}{2}\nabg_3 f +\ombg\, f +\frac 14 |f|^2\xibg+\frac 12(f\cdot \xibg)f.
\eeaa
Note that one can easily verify $\nab_3 f=\nabg_3 f$, as   $e_3$ is geodesic.  Since $\ombg, \xibg$ vanish we deduce that $f$ must verify the  equation $0= \etag -\frac{1}{2}\nab_3 f$.  
\begin{definition}
\lab{def:PT-frame}
 The  PT frame  $\et=e$   is  defined  
  by the transformation formula  
  \bea
  \lab{eq:PT-frame}
    e_4=\eg_4- f^a \,  \eg_a+\frac 14|f|^2 \eg_3,\quad e_a=\eg_a-\frac 12 \,  f _a \eg_3,  \quad e_3=\eg_3
\eea
with  $f$ the unique solution of  the equation  
\bea
\nab_3  f = 2  \etag , \qquad  f\Big|_{H_{-1}}=0.
\eea
\end{definition}
%%%%%%%%%%%

\begin{proposition}
\lab{Prop:PTframe-properties}
The PT frame  defined   above  verifies  the following properties:

\begin{enumerate}
\item  We have 
\bea
\lab{eq:PTframe-property1}
 \omb=\, \xib=0,\qquad  \etab=- \zeta, \qquad  \atrchb=0.
\eea
 \item We have 
 \bea
 \lab{eq:PTframe-property2}
    \nab_3 f=-\frac 12 \trchb f+2\zeta -\chibh\cdot f.
\eea

\end{enumerate}
\end{proposition}
\begin{proof}
To check \eqref{eq:PTframe-property1}  we start  from the fact that
  $e_3=\eg_3$ is  geodesic, i.e.   $\,\omb=0, \, \xib=0$.    Note that   $e_a(\ub)=\eg_a(\ub)-\frac 1 2 f_a  e_3(\ub)=0$. Hence     $e_a$ are tangent to  the level surfaces of $\ub$ 
 and so is the commutator  $[e_1, e_2]$. Thus $\atrchb= \in_{ab} g(D_a e_3, e_b)= - \frac 1 2  \in_{ab} g(e_3, [e_a, e_b])=   0$. 
  In view of the transformation formulas for $\ze$ and $\etab$  in Lemma \ref{Proposition:transformationRicci}  we  easily check  that  we also have  $ \ze+\etab= \zeg+ \etabg =0$.

 To check \eqref{eq:PTframe-property2}   we use  the inverse  transformation formulas,               corresponding  to $e \to \eg$. Thus  $\zeg = \ze   -\frac{1}{4}\trchb f -\frac{1}{2}\chibh\c f $.  Since $\zeg=\etag$  and $\etag = \frac{1}{2}\nab_3 f$    we  deduce that  $\frac{1}{2}\nab_3 f  =\ze   -\frac{1}{4}\trchb f -\frac{1}{2}\chibh\c f $ as stated.
       \end{proof}

\subsection{Null structure  and Bianchi equations in PT frame}

%%%%%%%%%%%%%%%%%%
\begin{proposition}
\label{Prop:transport.e_3}
Under the  ingoing PT frame the null structure  equations in the  incoming  direction $e_3$  take the form:
\beaa
\nab_3\trchb&=&-|\chibh|^2-\frac 1 2 \big( \trchb\big)^2,\\
\nab_3\chibh&=&-\trchb\,  \chibh-\aa,
\\
\nab_3\trch
&=& -\chibh\c\chih -\frac 1 2 \trchb\trch%+ {   2   \div\eta}+ 2 |\eta|^2
+ { {2\rho}},\\
\nab_3\atrch&=& -\chibh\wedge\chih-\frac 12\trchb\atrch%+2\curl\eta%
-2\dual\rho,\\
\nab_3\chih
&=&-\frac 1 2 \big( \trch \chibh+\trchb \chih\big)+\frac 12\dual\chibh\atrch,
\\
\nab_3 \zeta&=&-\chibh\c \zeta-\frac 12 \tr\chib \zeta-\bb,\\
\nab_3\om &=&|\ze|^2+ \rho,\\
\nab_3\xi&=&\chih\cdot \zeta+\frac 12\tr\chi \zeta -\frac 12\atrch \dual\zeta+\b.
\eeaa
We also have the equation of $\wtb:=\trchb+\frac{2}{|s|}$
\begin{equation*}
    {\nab_3 \widecheck{\trchb}+\trchb\widecheck{\trchb}=\frac 12(\wtb)^2-|\chibh|^2.}
\end{equation*}
The Bianchi equations take the form
    \beaa
    \nab_3\a-  \nab\hot \b&=&-\frac 1 2 \trchb\a+
  \ze\hot \b - 3 (\rho\chih +\rhod\dual\chih),\\
\nab_4\beta - \div\a &=&-2(\trch\beta-\atrch \dual \b) - 2  \om\b +\a\c  \ze + 3  (\xi\rho+\dual \xi\rhod),\\
     \nab_3\b+\div\varrho&=&-\trchb \b+2\bb\c \chih%+3 (\rho\eta+\rhod\dual \eta)
     ,\\
 \nab_4 \rho-\div \b&=&-\frac 3 2 (\trch \rho+\atrch \rhod)-\ze\c\b-2\xi\c\bb-\frac 1 2 \chibh \c\a,\\
   \nab_4 \rhod+\curl\b&=&-\frac 3 2 (\trch \rhod-\atrch \rho)+\ze\c\dual \b-2\xi\c\dual \bb+\frac 1 2 \chibh \c\dual \a, \\
     \nab_3 \rho+\div\bb&=&-\frac 3 2 \trchb \rho  +\ze  \c\bb-\frac{1}{2}\chih\c\aa,
 \\
   \nab_3 \rhod+\curl\bb&=&-\frac 3 2 \trchb \rhod+\ze \c\dual \bb-\frac 1 2 \chih\c\dual \aa,\\
     \nab_4\bb-\div\varoc&=&-(\trch \bb+ \atrch \dual \bb)+ 2\om\,\bb+2\b\c \chibh
    +3 (\rho\zeta-\rhod\dual \zeta)-    \aa\c\xi,\\
     \nab_3\bb +\div\aa &=&-2\trchb\,\bb +2\aa\c\ze ,\\
     \nab_4\aa+ \nab\hot \bb&=&-\frac 1 2 \big(\trch\aa {+} \atrch\dual \aa)+4\om \aa+
 5\ze\hot \bb - 3  (\rho\chibh -\rhod\dual\chibh).
\eeaa
Here,
\beaa
\div\varo&=&- (\nab\rho+\dual\nab\rhod),\\
\div\varoc&=&- (\nab\rho-\dual\nab\rhod).
\eeaa

 \begin{proof}
 Immediate consequence of  Propositions 2.2.5  and  2.2.6 in \cite{GKS} by using the  vanishing of  $\xib,\omb,\eta,\atrchb, \etab+\ze$.  The equation of $\wtb$ follows from  the one for 
 $\trchb $, $e_3(s)=1$, $s<0$ and direct computations.
 \end{proof}
\end{proposition}

\subsection{Commutation lemma}
We rely on the    general commutation     Lemma, see  section   2.2.7 in \cite{GKS},  to derive the following.
\begin{lemma}
 \lab{le:comm-PTframe1}
 With respect to the PT  frame  we have, for a general  $k$-horizontal  tensorfield $\psi_A=\psi_{a_1\cdots a_k}$,
  \bea
  \label{comm:PTframe1}
  \bsplit
    \,    [\nab_3,\nab_b]\psi_A&=-\chib_{bc}\nab_c\psi_A-\zeta_b\nab_3\psi_A-\sum_{i=1}^k \in_{a_i c}\dual\bb_b \psi_{a_1\cdots\, \, \, \cdots a_k}^{\quad\, \, \, \, c},\\
    \,  [\nab_3,\nab_4]\psi_A&=-2\etab_b \nab_b \psi_A+2\sum_{i=1}^k \in_{a_i b}\dual\rho\, \psi_{a_1\cdots\, \, \, \cdots a_k}^{\quad\, \, \, \, c}-2\om\nab_3 \psi_A,\\
     \,   [\nab_4,\nab_b]\psi_A&=-\chi_{bc}\nab_c \psi_A+\sum_{i=1}^k \Big(\chi_{b a_i}\etab_c-\chi_{bc}\etab_{a_i}\Big)\psi_{a_1\cdots\, \, \, \cdots a_k}^{\quad\, \, \, \, c}+\xi_b\nab_3\psi_A\\
            &\quad +\sum_{i=1}^k \Big(\chib_{b a_i}\xi_c-\chib_{bc}\xi_{a_i}+\in_{a_i c}\dual\b_b\Big)\psi_{a_1\cdots\, \, \, \cdots a_k}^{\quad\, \, \, \, c}.
            \end{split}
      \eea
      Moreover
      \bea
       \label{comm:PTframe2}
      \bsplit
         \,   [\nab_a,\nab_b]\psi_A&=\frac 12\atrch \nab_3\psi_A   \in_{ab}        +\Kh   \,   \sum_{i=1}^k \in_{a_i c} \left(g_{a_i a}g_{cb}-g_{a_i b}g_{ca}\right)  \psi_{a_1\cdots\, \, \, \cdots a_k}^{\quad\, \, \, \, c},\\
          \Kh&: =-\frac 14\trch\trchb+\frac 12 \chih\cdot\chibh-\rho.
          \end{split}
      \eea
         \end{lemma}
\begin{proof}
The commutation formulas  \eqref{comm:PTframe1}  follow  immediately   from   Lemma  2.2.7 in \cite{GKS} while \eqref{comm:PTframe2}  follows from  Proposition 2.1.45 in \cite{GKS}.
In both cases  we take into account  the  vanishing of  the quantities  $\xib,\omb,\, \eta,\atrchb,\, \etab+\ze$ in  our  PT frame.
\end{proof}

\section{Precise version of  the Main Theorem}
\label{section:Main Thm-precise}
Throughout the remaining of the paper, we use $\{e_3,e_4,e_a\}$ to denote the PT frame, and $\{\eg_3,\eg_4,\eg_a\}$ to denote the PG frame. We may also denote the $\eg$ frame simply by $e'$. We  denote the corresponding horizontal derivative $\nab$ and $\nabg$ (or $\nab'$).  We  shall also denote     $\d=(\nab,\nab_3)$.

%%%%%%%%%

\subsection{Main Norms}

%%%%%%%%%%%%%%%%%%

We    introduce our basic   integral norms  on $\MM$.     All Ricci and curvature coefficients  are  defined  with respect to the   PT frame   but may be integrated along the $S(u, s)$ spheres of the    associated geodesic foliation. Thus, for example, we define
\bea
\bsplit
\RR^S_{k} &:=  \frac{|s|^\frac 72}{\delta^\frac 52 a}||s^k\d^k\ab||_{L^2(S_{\ub,s})}+\frac{|s|^3}{\delta^2 a^\frac 32}||s^k\d^k\bb||_{L^2(S_{\ub,s})}+\frac{|s|^2}{\delta a}||s^k\d^k(\rho,\rhod)||_{L^2(S_{\ub,s})}\\
& +\frac {|s|}{a^\frac 12} ||s^k\d^k\b||_{L^2(S_{\ub,s})}+\frac{1}{\delta^{-1} a^\frac 12} ||s^k\d^k\a||_{L^2(S_{\ub,s})}
\end{split}
\eea
Also, with $ \trchbc:= \trchb+\frac{2}{|s|}$,
\bea
\bsplit
\mathcal O^S_{k}&:=\frac{1}{a^\frac 12}||s^k\d^k\chih||_{L^2(S_{\ub,s})}+||s^k\d^k\trch||_{L^2(S_{\ub,s})}+\frac{|s|^\frac 12}{\delta^\frac 12 a^\frac 12}||s^k\d^k\om||_{L^2(S_{\ub,s})}\\
& +\frac{|s|^{-\frac 12}}{\delta^{-\frac 12}a^\frac 12}||s^k\d^k\xi||_{L^2(S_{\ub,s})}+\frac{|s|}{\delta a^{\frac 12}}||s^k\d^k(\zeta,\chibh,\widecheck{\trchb})||_{L^2(S_{\ub,s})}+\frac{1}{\delta a^\frac 12}||s^k\d^k f||_{L^2(S_{\ub,s})}
\end{split}
\eea
along with a few $L^\infty$ norms
\begin{equation*}
    \OO_{k,\infty}^S:=\frac{|s|^2}{\delta a^\frac 12}||s^k \d^k (\chibh,\wtb,s^{-1} f)||_{L^\infty(S_{\ub,s})}.
\end{equation*}
We also define the energy type norms ($\Sigma_{\tau;\ub}$  refers to  the part of $\Sigma_\tau$ in the past of $\Hb_{\ub}$ and in the future of $\Hb_0\cup H_{-1}$)
\begin{equation*}
    \begin{split}
    \mathcal R_{k,2}&=\de^\frac 12 a^{-\frac 12}\left(||s^k\d^k\a||_{L^2(\Sigma_{\tau;\ub})}+||s^k\d^k\b||_{L^2(\underline H_{\ub})}\right)+\de^{-\frac 12} a^{-\frac 12}||s(s^k\d^k)(\rho,^*\rho)||_{L^2(\underline H_{\ub})}\\
    &  +\de^{-\frac 32} a^{-1} ||s^2 s^k\d^k\bb||_{L^2(\underline H_{\ub})}+\de^{-\frac 52} a^{-\frac 32}||s^3s^k\d^k\underline\a||_{L^2(\underline H_{\ub})}.
    \end{split}
\end{equation*}
We also use $\mathcal R_k[\psi]$ to denote the $\psi$-part of the $\RR_{k,2}$ norm, e.g., we denote $\mathcal R_k[\b]:=\delta^\frac 12 a^{-\frac 12}||s^k\nab^k\b||_{L^2(\hbub)}$.

We also make  use of the compound norms
\bea
\lab{eq:compundnorms}
  \o_{\le N}:=\sup_{k\leq N}\o_k^S, \qquad \RR_{\le N}:=\sup_{k\leq N}\RR_{k,2}, \qquad \RR^S_{\le N-1}=\sup_{k\leq N-1}\RR_{k}^S
  \eea
or simply $\OO$ and $\RR$  when there is no possible confusion.

\subsection{Initial data on \texorpdfstring{$H_{-1}$}{H(-1)}}
First, we note that the PT  frame\footnote{Recall that the PT  and   the PG  frames coincide on $H_{-1}$,  see Remark \ref{remark:PG=PT-onH_0}.}   on $H_{-1}$    coincides with the double null frame used in previous works.  One can thus easily compare  our  conditions with those of   \cite{Chr1} and \cite{AnLuk}.

\begin{proposition}
\label{Prop:InitailData-H_{-1}}
Assume that  the   short pulse condition \eqref{eq:chih_0-upperbound} holds  true for  some $N_0\ge 9$ and that the data on $H_{-1}\cap \Hb_0$ is Minkowskian. Then on $H_{-1}$,   as well as  in a local existence region\footnote{Note that 
$\atrch=0$ initially on $H_{-1}$   but not  in a non-trivial local existence region in the PT  frame.}, for   all  $i\leq N$, with  $\d=(\nab,\nab_3)$.
\bea
\bsplit
  \d^i(\xi,\chih)&\sim a^\frac 12,\quad \d^i(\trch,\atrch)\sim 1,\quad \d^i\om\sim \delta^\frac 12 a^\frac 12,\quad \d^i(\zeta, \chibh, \wtb)\sim \delta a^\frac 12\\
   \d^i\a&\sim \delta^{-1} a^\frac 12,\quad \d^i\b \sim a^\frac 12,\quad \d^i(\rho,\rhod)\sim \delta a,\quad \d^i\bb\sim \delta^2 a,\quad \d^i\aa\sim \delta^3 a^\frac 32,
\end{split}
\eea
\end{proposition}
\begin{proof}

To start with one can deduce, using  the analogues of  Proposition \ref{Prop:transport.e_3} in the $e_4$ direction\footnote{The $e_4$-equations on $H_{-1}$, where the foliation is geodesic, are similar to the ones in Proposition \ref{Prop:transport.e_3} by the substitution $e_3\to e_4$, $\chi\to \chib$, $\chib\to\chi$, $\xi\to \xib$, $\om\to \omb$, $\zeta\to -\zeta$ with potential loss of derivatives (this is like a PG frame rather than a PT frame), which does not matter on $H_{-1}$. Similar for curvature components and note that $\rhod\to -\rhod$.},  the  bounds
\beaa
\bsplit
  \chih&\sim a^\frac 12,\quad \trch\sim 1,\quad   \zeta\sim \delta a^\frac 12,\quad \chibh\sim \delta a^\frac 12,\quad \trchb+2\sim \delta a^\frac 12,\\
        \a &\sim \delta^{-1} a^\frac 12,\quad \b \sim a^\frac 12,\quad (\rho,\rhod)\sim \delta a,\quad \bb\sim \delta^2 a,\quad \aa\sim \delta^3 a^\frac 32.
\end{split}
\eeaa
We   can then show, using the commutation  Lemma \ref{le:comm-PTframe1}   that the same asymptotic conditions   hold true   for   the angular derivatives  $\nab$ of these components. 
The same bounds for the $\nab_3$          derivatives hold also true- they can be easily deduced from transport equations in Proposition  \ref{Prop:transport.e_3}. Indeed, all   quantities  except $\aa$,  verify  a $\nab_3$ equation.
Estimates   for $\nab_3\ab$    can be  derived by   integrating   on $H_{-1}$ of   the equation     for $\nab_4(\nab_3 \aa)  $ obtained      by commuting        the $\nab_4\ab$   Bianchi identity   with $\nab_3$. Schematically, 
\begin{equation*}
    \nab_4 \nab_3\ab+\nab\hot\nab_3\bb=[\nab_4,\nab_3]\ab+[\nab\hot,\nab_3]\bb+\nab_3(\Gac_b\cdot \ab)+\nab_3(\Gac_g \cdot (\bb,\rho,\rhod)).
\end{equation*}
For details  we  refer the reader to \cite{Chr1}.
\end{proof}

\subsection{Main Theorem (second version)}

\begin{theorem}
\label{theorem:main'}
    Consider the characteristic initial value problem described above, with the data on $\Hb_0$ Minkowskian or a perturbation satisfying $\RR_{k,2}|_{\ub=0}\lesssim 1$ for $k\leq N$.
    \begin{enumerate}
    \item  If \eqref{eq:chih_0-upperbound} holds, then the spacetime can be extended to $\MM(\delta,a;-\frac 18 a\delta)$ such that  the following estimate hold true for all $k\le N$, with a sufficiently large constant $C>0$,
    \bea
    \lab{eq:MaiThem-etsimate}
    \RR_{k,2}\leq C,\qquad  \o^S_{k}\leq C.
\eea
   \item   Moreover, if \eqref{eq:chih_0-lowerbound} also holds, then $S_{\delta,-\frac 14 a\delta}$ is a trapped surface.
   \end{enumerate}
\end{theorem}

\medskip

We prove the theorem by a continuity argument  based  on the following bootstrap assumption.\\
\NI
{\bf Bootstrap assumptions.} Assume that  for  some $\tau^*$ with $-1<\tau^*\leq -\frac 18 a\delta$, the following bounds hold  true for all $\ub$ and $s$ satisfying $\tau=\frac{1}{10}a\ub+s\leq \tau^*$, and for a sufficiently large constant $C_b$ to be chosen,
\bea
\label{eq:bootstrap}
\bsplit
    \o_{k,\infty}^S(\ub,s)&\leq C_b,\qquad \mbox{for all}\,\, k\leq [N/2]+1,\\
    \RR^S_{k}(\ub,s)&\leq C_b,\qquad \mbox{for all}\,\,  k\leq N-1,\\
     \o^S_{k}(\ub,s)+\RR_{k,2}(\ub)&\leq C_b,\qquad \mbox{for all}\,\,  k\leq N.
    \end{split}
\eea
The existence of such $\tau^*$ is ensured by  a standard, characteristic   local wellposedness  result, see for example \cite{Luk}.
We  improve the bootstrap assumption, summarized   in the sequence of steps below,    showing   that  the constant $C_b$ can be replaced by a universal constant depending 
only on the initial data and  the size of  $-\frac 18 a\delta$. The local existence result can  then be invoked  to  extend the existence region  to the whole  $\MM(a,\delta;-\frac 18 a\delta)$  while preserving the 
same bounds.

\subsubsection{Main Steps}
\begin{enumerate}
    \item In Section \ref{Section:Integrating-Ricci}, we integrate the $e_3$-transport equations in Proposition \ref{Prop:transport.e_3},  including the equation \eqref{eq:PTframe-property2}  for $f$,  and derive  $L^2$-estimates on the spheres $S_{\ub, s}$.
    \item   To pass from $L^2$ to $L^\infty$ estimates  we need to rely on a version of the Sobolev inequalities  which holds true for the   non-integrable    PT frame.
    This is done in   Section \ref{Section:Non-integrable-Sobolev}   by  going back and forth between the  PT frame and  the   integrable PG  frame\footnote{Where we can rely on the standard Sobolev inequalities in the geodesic frame.}.
    \item In Section \ref{Section:energy-estimates}, we derive the spacetime energy estimate for the null curvature components and close the bootstrap argument.
\end{enumerate}
\begin{remark}
\lab{remark:size-tau_*}  We note that since $0\leq \ub \leq \delta$, the size of $|\tau|$ and $|s|$ are always comparable. In particular, we always have $|s|\geq \frac 18 a\delta$. The reason to introduce $\tau$ is to give an achronal boundary of the spacetime to derive the energy estimates.
\end{remark}

\section{Sobolev estimates in a non-integrable frame}\label{Section:Non-integrable-Sobolev}
Recall that we denote by $\SS$  the horizontal structure given by the PG  frame and by $\HH$ the  one of the PT frame. For simplicity of notation we   use   $'$  rather than $\,^{(g)}$ to denote the quantities associated to the PG  frame  $\SS$.

\begin{lemma}
\lab{Lemma:Sobolev}
    Suppose that the bootstrap assumption \eqref{eq:bootstrap} holds. Then for an $\HH$-horizontal covariant tensor $\psi$, we have the estimate\footnote{Throughout this work, the implicit constant implied by the symbol ``$\lesssim$" is independent of bootstrap constants $C_b$ and $\OO$, $\RR$.}
    \begin{equation}\label{eq:non-integrable-Sobolev-estimate}
        ||\psi||_{L^\infty (S_{\ub,s})}\lesssim \sum_{i\leq 2} ||s^i \nab^i \psi||_{L^2(S_{\ub,s})}+a^{-\frac 12} ||(s\d)^{\leq 2} \psi||_{L^2(S_{\ub,s})}.
    \end{equation}
    for any $S_{\ub,s}$ in $\MM(\delta,a;\tau^*)$.
    The right hand side can also simply be replaced by $\sum_{i\leq 2}||s^i\d^i\psi||_{L^2(S_{\ub,s})}$.
\end{lemma}
\begin{proof}
    Given an $\HH$-horizontal covariant tensor $\psi_A=\psi_{a_1\cdots a_k}$, we define the $\SS$-horizontal tensor $\tilde\psi_A=\tilde \psi_{a_1\cdots a_k}$ so that
\begin{equation*}
    \tilde\psi_{a_1\cdots a_k}:=\tilde\psi(e_{a_1}',\cdots,e_{a_k}')=\psi_{a_1\cdots a_k}
\end{equation*}
for any $a_1,\cdots,a_k$. To apply the Sobolev estimate, we wish to control $\nab'\nab'\tilde \psi$. We first compute $\nab'\psi$, which is an $\SS$-horizontal covariant $(k+1)$-tensor:
\begin{equation*}
\begin{split}
    \nab_b'\tilde\psi_{a_1\cdots a_k}&=e_b'(\tilde\psi_{a_1\cdots a_k})-\tilde\psi(D_{e_b'}e_{a_1}',\cdots,e_{a_k}')-\cdots-\tilde\psi (e_{a_1}',\cdots,D_{e_b'}e_{a_k}')\\
    &=e_b'(\psi_{a_1\cdots a_k})-\tilde\psi(D_{e_b'}e_{a_1}',\cdots,e_{a_k}')-\cdots-\tilde\psi (e_{a_1}',\cdots,D_{e_b'}e_{a_k}')\\
    &=(D_{e_b'}\psi)(e_{a_1},\cdots,e_{a_k})+\psi(D_{e_b'}e_{a_1},\cdots,e_{a_k})+\cdots+\psi(e_{a_1},\cdots,D_{e_b'}e_{a_k})\\
    &\quad -\tilde\psi(D_{e_b'}e_{a_1}',\cdots,e_{a_k}')-\cdots-\tilde\psi (e_{a_1}',\cdots,D_{e_b'}e_{a_k}')\\
    &=\nab_{e_b'}\psi_{a_1\cdots a_k}+\sum_{i=1}^k\left(g(D_{e_b'}e_{a_i},e_c)-g(D_{e_b'}e_{a_i'},e_c')\right)\psi_{a_1\cdots c\cdots a_k}\\
    &=\nab_{e_b'}\psi_{a_1\cdots a_k}+\sum_{i=1}^k\left(-\frac 12 f_{a_i}\chib_{bc}\psi_{a_1\cdots c\cdots a_k}+\frac 12 f_c \chib_{a_i b} \psi_{a_1\cdots c \cdots a_k}\right)=:(\nab\tilde\psi)_{ba_1\cdots a_k},
    \end{split}
\end{equation*}
where we used $\xib=0$ in the last step. Here we define $\nab\tilde\psi$ as an $\HH$-horizontal tensor (this is simply a notation). 
%\Blue{(The difference between $e_{a_k}'$ and $e_{a_k}$ is in $e_3$ direction, so this type of difference gives a zero component of $\psi$.)} 
Then, applying the calculation above with $\tilde\psi$ replaced $\nab'\tilde\psi$, we obtain
\begin{equation*}
\begin{split}
    \nab_b'\nab_c'\tilde\psi_{a_1\cdots a_k}=\nab_{e_b'}(\nab\tilde\psi)_{c a_1\cdots a_k}+\sum_{i=0}^k(\nab\tilde\psi)_{a_0 a_1\cdots d\cdots a_k}\cdot \left(-\frac 12 f_{a_i}\chib_{bd}+\frac 12 f_d \chib_{a_i b}\right),
\end{split}
\end{equation*}
where $a_0:=c$. For the first term, by the definition of $\nab\tilde\psi$, we have
\begin{equation*}
\begin{split}
    \nab_{e_b'}(\nab\psi)_{c a_1\cdots a_k}&=\nab_{e_b'}\nab_{e_c'}\psi_{a_1\cdots a_k}+\sum_{i=1}^k \nab_{e_b'}\left(\Big(-\frac 12 f_{a_i}\chib_{bd}+\frac 12 f_d\chib_{a_i c}\Big)\psi_{a_1\cdots d\cdots a_k}\right)\\
    &=(\nab_b+\frac 12 f_b \nab_3)(\nab_c+\frac 12 f_c \nab_3)\psi_{a_1\cdots a_k}+\nab (f\cdot \chib\cdot \psi)\\
    &=\nab_b \nab_c \psi_{a_1\cdots a_k}+\d f\cdot \d\psi+f\cdot \d^{\leq 2}\psi+\nab(\Gac_g\cdot\psi)
    \end{split}
\end{equation*}
using that $|s|^{-1}f\in\Gac_g^1:=\{\chibh,\wtb,|s|^{-1}f\}$, and hence $f\cdot\chib=\Gac_g^1\cdot f+|s|^{-1}f\in \Gac_g^1$.

Therefore, using the bootstrap assumption $|s\nab\Gac_g^1|,|\Gac_g^1|\lesssim \o \delta a^\frac 12 |s|^{-2}$ for $\Gac_g^1$ and the Sobolev estimate on the sphere $S_{\ub,s}$ (see appendix for the proof), we derive, using $|s|\geq \frac 18 a\delta$,
\begin{equation*}
    \begin{split}
    ||\psi||_{L^\infty(S_{\ub,s})}&=||\tilde\psi||_{L^\infty(S_{\ub,s})}\lesssim ||(s\nab')^{\leq 2}\tilde\psi||_{L^2(S_{\ub,s})}\\
    &\lesssim ||s^2\nab_{e_b'}(\nab\tilde\psi)||_{L^2(S_{\ub,s})}+||s\nab\tilde\psi||_{L^2(S_{\ub,s})}||sf\cdot\chib||_{L^\infty(S_{\ub,s})}\\
    &\lesssim ||s^2\nab^2\psi||_{L^2(S_{\ub,s})}+||(s\nab)^{\leq 1}\psi||_{L^2(S_{\ub,s})} ||s(s\nab)^{\leq 1} \Gac_g||_{L^\infty(S_{\ub,s})}\\
    &\quad +||f||_{L^\infty(S_{\ub,s})}||(s\d)^{\leq 2}\psi||_{L^2(S_{\ub,s})}+||s\d f||_{L^\infty(S_{\ub,s})}||s\d\psi||_{L^2(S_{\ub,s})}\\
    &\lesssim \sum_{i=0}^2 ||s^i\nab^i \psi||_{L^2(S_{\ub,s})} (1+\OO \delta a^\frac 12 |s|^{-2})+\OO \delta a^\frac 12 |s|^{-2} ||(s\d)^{\leq 2}\psi||_{L^2(S_{\ub,s})}\\
    &\lesssim \sum_{i=0}^2 ||s^i\nab^i \psi||_{L^2(S_{\ub,s})}+a^{-\frac 12} ||(s\d)^{\leq 2}\psi||_{L^2(S_{\ub,s})},
    \end{split}
\end{equation*}
and the result follows.
\end{proof}
With this estimate, each bootstrap bound on $L^2(S_{\ub,s})$ norm implies a lower-order $L^\infty(S_{\ub,s})$ bound of the same quantity. We will make use of these $L^\infty$ bounds without mentioning the use of the Sobolev lemma.

%%%%%%%%%%%%

\section{Estimate of Ricci coefficients in PT frame}\label{Section:Integrating-Ricci}

%%%%%%%%%%%%%%%%%

We denote $\Gamma$ as all possible Ricci coefficients in the PT frame. We also  set, with  $f$   introduced in Definition \ref{def:PT-frame}   and   verifying   \eqref{eq:PTframe-property2},
\begin{equation*}
    \Gac:=\{\xi,\chih,\trch,\atrch,\om,\zeta,\chibh,\widecheck{\trchb},|s|^{-1}f\}=  \Gac_a\cup \Gac_b\cup\Gac_g
\end{equation*}
 where
\bea
    \Gac_a=\{\xi,\chih\},\quad \Gac_b=\{\trch,\atrch,\om\},\quad \Gac_g=\{\zeta,\chibh,\wtb,|s|^{-1}f\}.
\eea
\begin{remark}Note that $\nab\trchb=\nab\widecheck{\trchb}+\nab(2s^{-1})=\nab\widecheck{\trchb}-2s^{-2}f=s^{-1}(s\nab\wtb-2s^{-1}f)$, so we see that while $\trchb$ is not in $\Gac_g$, $s\nab\trchb$ is schematically $s\nab\Gac_g+\Gac_g$ (higher orders are also similar).
\end{remark}

%%%%%%%%%%%%%

\subsection{Integrating the model transport equation}

%%%%%%%%%%%%%%%%%%

We  rely on the  following  weighted integration lemma. This is similar to Proposition 5.5 in \cite{AnLuk}.
\begin{lemma}
\lab{Lemma:transport-3}
    Suppose that the bootstrap assumptions \eqref{eq:bootstrap} holds true in $\MM(\de, a, \tau_*)$.  Then for a $\HH$-horizontal covariant tensor field satisfying the equation
    \begin{equation*}
        \nab_3\psi+\lambda\trchb \psi=F,
    \end{equation*}
    we have,  for $\lambda_1=2(\lambda-\frac 12)$ and $s\leq -\frac 18 a\delta$,
    \begin{equation}\label{eq:integrate-model-zero-order}
        |s|^{\lambda_1}||\phi||_{L^2(S_{\ub,s})}\lesssim ||\phi||_{L^2(S_{\ub,-1})}+\int_{-1}^{s} |s'|^{\lambda_1} ||F||_{L^2(S_{\ub,s'})}\, ds'.
    \end{equation}
    We also have the higher-order estimates:
    \begin{equation}\label{eq:integrate-model-higher-order}
    \begin{split}
        ||s^{2\lambda-1+k}\d^k\psi||_{L^2(S_{\ub,s})}\lesssim ||s^{2\lambda-1+k}\d^k\psi||_{L^2(S_{\ub,-1})}+\sum_{i\leq k}\int_{-1}^s |s'|^{2\lambda-1+i} ||\d^i F||_{L^2(S_{\ub,s'})}\, ds'
    \end{split}
\end{equation}
\end{lemma}

\begin{proof}
Note that  $e_3(|\psi|^2)=2\psi\cdot \nab_3\psi$. %$=e_3(g^{\mu\nu}\psi_\mu\psi_\nu)=2g^{\mu\nu} D_3\psi_\mu \psi_\nu=(-\frac 12 ((e_3)^\mu (e_4)^\nu+(e_4)^\mu (e_3)^\nu)+(e_a)^\mu (e_a)^\nu+(e_b)^\mu (e_b)^\nu)D_3\psi_{\mu} \psi_\nu=D_3\psi_a \psi_a+D_3\psi_b \psi_b=2\psi\cdot \nab_3\psi$. 
    We  can thus  make use  of the following    variation formula  for scalar functions $\phi$
    \begin{equation*}
        \pa_s\int_{S_{\ub,s}} \phi=\int_{S_{\ub,s}} e_3\phi+\trchb\phi
    \end{equation*}
    since $e_3=\pa_s$. Letting $\phi=|s|^{2\lambda_1}|\psi|^2$, we have (note that $s<0$)
       \begin{equation*}
        \begin{split}
            \Big|\pa_s \int_{S_{\ub,s}} |s|^{2\lambda_1} |\psi|^2\Big|&=\Big|\int_{S_{\ub,s}} -2\lambda_1|s|^{2\lambda_1-1}|\psi|^2+2|s|^{2\lambda_1}\psi\cdot \nab_3\psi+|s|^{2\lambda_1}\trchb|\psi|^2\Big|\\
            &=\Big|\int_{S_{\ub,s}} 2|s|^{2\lambda_1}\psi\cdot (-\lambda_1 |s|^{-1}\psi+\nab_3\psi)+|s|^{2\lambda_1}\trchb|\psi|^2\Big|\\
            &=\Big|\int_{S_{\ub,s}} 2|s|^{2\lambda_1} \psi\cdot (F-\lambda \trchb\psi+\lambda_1|s|^{-1}\psi)+|s|^{2\lambda_1}\trchb|\psi|^2\Big|\\
            &=|\int_{S_{\ub,s}} 2|s|^{2\lambda_1}\psi\cdot F-\lambda_1|s|^{2\lambda_1}\psi\cdot \wtb \psi\Big|\\
            &\leq 2\left(\int_{S_{\ub,s}} |s|^{2\lambda_1} |\psi|^2\right)^\frac 12 \left(\int_{S_{\ub,s}} |s|^{2\lambda_1} |F|^2\right)^\frac 12+\lambda_1 \o\delta a^\frac 12 |s|^{-2} \int_{S_{\ub,s}} |s|^{2\lambda_1} |\psi|^2,
        \end{split}
    \end{equation*}
    where we  made use   of the bootstrap assumption in the last step. Now, since\footnote{See Remark \ref{remark:size-tau_*} .} $|s|\geq \frac 18 a \delta$, so we have $\int_{-1}^s \o\delta a^\frac 12 |s'|^{-2} ds'\lesssim \OO a^{-\frac 12}\ll 1$, the estimate follows by integration using the  Gr\"{o}nwall  inequality. This finishes the proof of \eqref{eq:integrate-model-zero-order}.

    For the higher-order version, we need to commute the equation with $\d=(\nab,\nab_3)$. 
When we commute the equation with $\nab_3$, we have %(when $\lambda\neq 0$; the case $\lambda=0$ is trivial)
\begin{equation*}
    \nab_3\nab_3\psi+\lambda\trchb \nab_3\psi+\lambda\nab_3(\trchb)\psi=\nab_3 F,
\end{equation*}
so using $\nab_3\trchb=-|\chibh|^2-\frac 12(\trchb)^2$ and the original equation $\lambda \trchb \psi=F-\nab_3\psi$, we get
\begin{equation*}
    \nab_3(\nab_3\psi)+\lambda \trchb\nab_3\psi-\frac 12\lambda|\chibh|^2\psi+\frac 12 \trchb (\nab_3 \psi-F)=\nab_3 F,
\end{equation*}
i.e.,
\begin{equation*}
    \nab_3(\nab_3\psi)+\Big(\lambda+\frac 12\Big)\trchb\nab_3\psi=\nab_3 F+\frac 12 \trchb F+\frac 12 \lambda|\chibh|^2 \psi
\end{equation*}
The term $|\chih|^2\psi=s^{-1}(s|\chibh|^2)\psi$ and is of the type (in fact better) $s^{-1}\Gac_g \psi$.
Inductively, we get the schematic expression 
\begin{equation*}
    \nab_3 \nab_3^i \psi+\Big(\lambda+\frac i2\Big) \trchb \nab_3^i \psi=\nab_3^i F+\sum_{j=1}^{i} |s|^{-j}\nab_3^{i-j}(F+\Gac_g\cdot \psi).
\end{equation*}
The commutation with $\nab^i$, by the formula, gives
\begin{equation*}
    [\nab_3,\nab^i]\psi=-\frac i2 \trchb \nab^i\psi+\nab^{i-1}(\Gac_g\cdot  \d\psi)+\nab^{i-1}(\bb^*\cdot\psi)
\end{equation*}
hence
\begin{equation*}
    \nab_3\nab^i\psi+\Big(\lambda+\frac i2\Big)\trchb\nab^i\psi=\nab^i F+\nab^{i-1}(\Gac_g\cdot  \d\psi)+\nab^{i-1}(\bb^*\cdot\psi).
\end{equation*}
Therefore, to commute with $\d^i$, we can commute in either way for finite times, and get
\begin{equation*}
    \nab_3\d^i\psi+\Big(\lambda+\frac i2\Big) \trchb \d^i\psi=\sum_{j=0}^i |s|^{-j}\d^{i-j} F+\d^{\leq i-1}(\Gac_g\cdot  \d\psi)+\d^{\leq i-1}(\bb^*\cdot\psi).
\end{equation*}
Then, applying the integration lemma, we get (for simplicity, we denote here $S=S_{\ub,s'}$)
\begin{equation*}
\begin{split}
    ||&s^{2\lambda+i-1}\d^i \psi||_{L^2(S_{\ub,s})}\lesssim ||s^{2\lambda+i-1}\d^i\psi||_{L^2(S_{\ub,-1})}+\int_{-1}^s |s'|^{2\lambda+i-1} ||\d^i F||_{L^2(S)}\, ds'\\
    &\quad +\int_{-1}^s |s'|^{2\lambda+i-1} ||\d^i(\Gac_g \cdot \psi)||_{L^2(S)}+||\d^{i-1}(\bb^*\cdot\psi)||_{L^2(S)}\, ds'\\
    &\lesssim ||s^{2\lambda+i-1}\d^i\psi||_{L^2(S_{\ub,-1})}+\int_{-1}^s |s'|^{2\lambda+i-1} ||\d^i F||_{L^2(S)} ds'\\
    %Long formula 1
    &\quad+\int_{-1}^s |s'|^{2\lambda-1}\sum_{\substack{i_1+i_2=i\\ i_2\leq i/2}} ||s'^{i_1}\d^{i_1}\Gac_g||_{L^2(S)}||s'^{i_2}\d^{i_2}\psi||_{L^\infty(S)}+||s'^{i_2}\d^{i_2}\Gac_g||_{L^\infty(S)}||s'^{i_1}\d^{i_1}\psi||_{L^2(S)}\\
    %Long formula 2
    &\quad+\int_{-1}^s |s'|^{2\lambda}\sum_{\substack{i_1+i_2=i-1\\ i_2\leq i/2}} ||s'^{i_1}\d^{i_1}\bb||_{L^2(S)}||s'^{i_2}\d^{i_2}\psi||_{L^\infty(S)}+||s'^{i_2}\d^{i_2}\bb||_{L^\infty(S)}||s'^{i_1}\d^{i_1}\psi||_{L^2(S)}\\
    %lesssim
    &\lesssim ||s^{2\lambda+i-1}\d^i\psi||_{L^2(S_{\ub,-1})}+\int_{-1}^s |s'|^{2\lambda+i-1} ||\d^i F||_{L^2(S)}\, ds'\\
    &\quad+\int_{-1}^s \mathcal O\frac{\delta a^\frac 12}{|s'|}\cdot \sum_{i_2\leq i/2} ||s'^{2\lambda-1}s'^{i_2}\d^{i_2}\psi||_{L^\infty(S_{\ub,s'})}+\mathcal O\frac{\delta a^\frac 12}{|s'|^2} \sum_{i_1\leq i}||s'^{2\lambda-1+i_1}\d^{i_1}\psi||_{L^2(S)}\, ds'\\
    &\quad+\int_{-1}^s \RR^S_{\leq i-1}[\bb] \delta^2 a^\frac 32 |s'|^{-3}\cdot |s'|\sum_{i_2\leq i/2}  ||s'^{2\lambda-1+i_2}\d^{i_2}\psi||_{L^\infty(S_{\ub,s'})} ds'\\
    &\quad+\int_{-1}^s \RR^S_{\leq i-1}[\bb] \delta^2 a^\frac 32 |s|^{-4}\cdot |s'| \sum_{i_1\leq i-1}||s'^{2\lambda-1}s^{i_1}\d^{i_1}\psi||_{L^2(S_{\ub,s'})}\, ds'.
\end{split}
\end{equation*}
Then, using the non-integrable Sobolev estimate, together with the bootstrap assumption \eqref{eq:bootstrap}, we have
\begin{equation*}
\begin{split}
    ||s^{2\lambda+i-1}&\d^i \psi||_{L^2(S_{\ub,s})}\lesssim  ||s^{2\lambda+i-1}\d^i\psi||_{L^2(S_{\ub,-1})}+\int_{-1}^s |s'|^{2\lambda+i-1} ||\d^i F||_{L^2(S_{\ub,s'})}\, ds'\\
    &+\int_{-1}^s \big(C_b \frac{\delta a^\frac 12}{|s'|^2}+C_b \delta^2 a^\frac 32 |s'|^{-3}\big) \sum_{j\leq i}||s'^{2\lambda-1+j}\d^{j}\psi||_{L^2(S_{\ub,s'})}\, ds'.
\end{split}
\end{equation*}
Since \[\int_{-1}^{-a\delta}  C_b\frac{\delta a^\frac 12}{|s'|^3}+C_b\delta^2 a^\frac 32 |s'|^{-2} ds'\lesssim C_b\, a^{-\frac 12}\ll 1,\]we can sum up $i=1,\cdots,k$ and use Gr\"{o}nwall's lemma to get
\begin{equation*}
    \begin{split}
        ||s^{2\lambda-1+k}\d^k\psi||_{L^2(S_{\ub,s})}\lesssim ||s^{2\lambda-1+k}\d^k\psi||_{L^2(S_{\ub,-1})}+\sum_{i\leq k}\int_{-1}^s |s'|^{2\lambda-1+i} ||\d^i F||_{L^2(S_{\ub,s'})}\, ds'
    \end{split}
\end{equation*}
which finishes the proof of \eqref{eq:integrate-model-higher-order}.
\end{proof}
\begin{remark}
    From the proof, it is clear that the same estimate holds when $F$ is replaced by $F+\Gac_g\cdot {\psi}$ in view of the bootstrap bounds of $\Gac_g$. %with $\Gac_g$ satisfying the bound $|\Gac_g|\lesssim \o \delta a^\frac 12 |s|^{-2}$.
\end{remark}

\subsection{Estimate of Ricci coefficients}
\begin{remark}   We make systematic use  of the $e_3$ transport  equations of Proposition \ref{Prop:transport.e_3} to which we apply the transport Lemma 
\ref{Lemma:transport-3}.
Without  further notice  we   estimate weighted $L^2(S_{\ub, s})$  expressions  of the form  $  ||s^i\d^i(\psi_1\cdot\psi_2)||_{L^2(S_{\ub,s})}$  by
\begin{equation*}
    \sum_{\substack{i_1+i_2=i\\ i_2\leq i/2}} ||s^{i_1}\d^{i_1}\psi_1||_{L^2(S_{\ub,s})}||s^{i_2}\d^{i_2}\psi_2||_{L^\infty(S_{\ub,s})}+||s^{i_2}\d^{i_2}\psi_1||_{L^\infty(S_{\ub,s})}||s^{i_1}\d^{i_1}\psi_2||_{L^2(S_{\ub,s})}.
\end{equation*}
\end{remark}
%%%%%%%%%%%%
\begin{proposition}
    We have $||s^k\d^k\om||_{L^2(S_{\ub,s})}\lesssim \mathcal R_k[\rho]\, \delta^\frac 12 a^\frac 12 |s|^{-\frac 12}$, $k\leq N$.
\end{proposition}
\begin{proof}
We apply Lemma \ref{Lemma:transport-3}         to  the equation  $ \nab_3\, \om=|\zeta|^2+\rho$ 
and derive
\begin{equation*}
    \begin{split}
    ||\, |s|^{i-1}\d^i\om &||_{L^2(S_{\ub,s})}\lesssim  ||\, |s|^{i-1}\d^i\om||_{L^2(S_{\ub,-1})}+\int_{-1}^{s} |s'|^{i-1} ||\d^i(|\zeta|)^2||_{L^2(S_{\ub,s'})}+|s'|^{i-1} ||\d^i\rho||_{L^2(S_{\ub,s'})} ds' \\
    %1
    &\lesssim  0+\int_{-1}^{s} |s'|^{-1} \o \delta a^\frac 12 |s'|^{-2}\cdot \o\delta a^\frac 12 |s'|^{-1}\, ds'+\Big(\int_{-1}^s |s'|^{-4} ds'\Big)^\frac 12 ||s(s^i\d^i\rho)||_{L^2(\hbub)}\\
    &\lesssim  \mathcal O^2\delta^2 a |s|^{-3}+\RR_i[\rho]\delta^\frac 12 a^\frac 12 |s|^{-\frac 32}\\
    &\lesssim  (\RR_i[\rho]+\mathcal O^2 a^{-\frac 12})\delta^\frac 12 a^\frac 12 |s|^{-\frac 32}\lesssim \mathcal R_i[\rho]\, \delta^\frac 12 a^\frac 12 |s|^{-\frac 32}.%\qedhere
    \end{split}
\end{equation*}
\end{proof}

\begin{proposition}
\label{Prop:estimate-xi}
    We have $||s^k\d^k \xi||_{L^2(S_{\ub,s})}\lesssim \RR_k[\b]\, \delta^{-\frac 12} a^\frac 12 |s|^{-\frac 12}$, $k\leq N$.
\end{proposition}
\begin{proof}
We apply Lemma \ref{Lemma:transport-3}         to  the equation  for $\nab_3 \xi$ written schematically in the form
$\nab_3\xi=\chih\cdot \zeta+\Gac_b\cdot \zeta+\b$ to derive 
   \begin{equation*}
    \begin{split}
    ||\, |s|^{i-1}\d^i\xi||_{L^2(S_{\ub,s})}&\lesssim  ||\, |s|^{i-1}\d^i\xi||_{L^2(S_{\ub,-1})}+\int_{-1}^{s} |s'|^{i-1} ||\d^i((\chih,\Gac_b)\cdot\zeta)||_{L^2(S_{\ub,s'})}+|s'|^{i-1} ||\d^i\b||_{L^2(S_{\ub,s'})} ds' \\
    &\lesssim  0+\int_{-1}^{s} |s'|^{-1}\o a^\frac 12 |s'|^{-1}\cdot \o \delta a^\frac 12 |s'|^{-1}\, ds'
    +\Big(\int_{-1}^s |s'|^{-2} ds'\Big)^\frac 12 ||s^i\d^i\b||_{L^2(\hbub)}\\
    &\lesssim  \mathcal O^2 \delta a |s|^{-2}+\RR_i[\b]\, \delta^{-\frac 12} a^\frac 12 |s|^{-\frac 12}\\
    &\lesssim  (\RR_i[\b]+\mathcal O^2 a^{-\frac 12})\, \delta^\frac 12 a^\frac 12 |s|^{-\frac 32}\lesssim \mathcal R_i[\b]\, \delta^\frac 12 a^\frac 12 |s|^{-\frac 32}.
    \end{split}
\end{equation*}
\end{proof}

\begin{proposition}
\lab{Prop:estimate-ze}
    We have $||s^k\d^k \zeta||_{L^2(S_{\ub,s})}\lesssim \delta a^\frac 12+\RR_k[\bb]\, \delta^\frac 32 a|s|^{-\frac 32}$, $k\leq N$.
\end{proposition}
\begin{proof}
We apply  Lemma \ref{Lemma:transport-3} (with $\lambda=\frac 12$)    to the equation  for $\nab_3 \ze $,  written schematically in the form $\nab_3 \zeta+\frac 12\trchb\, \zeta=\Gac_g\cdot\zeta-\bb$, 
\begin{equation*}
    \begin{split}
    ||\, |s|^{i}\d^i\zeta &||_{L^2(S_{\ub,s})}\lesssim  ||\, |s|^{i}\d^i\zeta||_{L^2(S_{\ub,-1})}+\int_{-1}^{s} |s'|^{i} ||\d^i(\Gac_g\cdot\zeta)||_{L^2(S_{\ub,s'})}+|s'|^{i} ||\d^i\bb||_{L^2(S_{\ub,s'})} ds' \\
    &\lesssim  \delta a^\frac 12+\int_{-1}^{s} \o\delta a^\frac 12 |s'|^{-2}\cdot  \o \delta a^\frac 12 |s'|^{-1}\, ds'+\Big(\int_{-1}^s |s'|^{-4} ds'\Big)^\frac 12 ||s^2 s^i\d^i\bb||_{L^2(\hbub)}\\
    &\lesssim  \delta a^\frac 12+\mathcal O^2 \delta^2 a |s|^{-2}+\RR_i[\bb]\delta^{\frac 32} a |s|^{-\frac 32}\\
    &\lesssim  \delta a^\frac 12+(\RR_i[\bb]+\mathcal O^2 a^{-\frac 12})\delta^\frac 32 a |s|^{-\frac 32}\lesssim \delta a^\frac 12+\mathcal R_i[\bb]\, \delta^\frac 32 a |s|^{-\frac 32}.
    \end{split}
\end{equation*}
\end{proof}
\begin{remark}
    This is slightly better than the $\delta a^\frac 12 |s|^{-2}$ size of the bootstrap assumption.  This turns out to be useful to avoid a logarithmic loss in the estimate of the frame transformation $f$.
\end{remark}

\begin{proposition}\label{prop:L2-estimate-wtb,chibh}
    We have $||s^k\d^k (\widecheck{\trchb},\chibh)||_{L^2(S_{\ub,s})}\lesssim \delta a^\frac 12 |s|^{-1}$.
\end{proposition}
\begin{proof}
We  apply  Lemma \ref{Lemma:transport-3}  (with $\lambda=1$)   to the equations
\beaa
    \nab_3 \trchbc +\trchb\trchbc&=&\frac 12(\wtb)^2-|\chibh|^2,\qquad 
    \nab_3 \chibh+\trchb\, \chibh=-\underline\a.
\eeaa
%\begin{equation*}
%    \nab_3\atrchb=-\atrchb\trchb.
%\end{equation*}
\begin{equation*}
    \begin{split}
        ||s^{1+i}\d^i(\chibh,\wtb)||_{L^2(S_{\ub,s})}&\lesssim ||s^{1+i}\d^i(\chibh,\wtb)||_{L^2(S_{\ub,-1})}+\int_{-1}^s |s'|^{1+i} ||\d^i (\Gac_g\cdot\Gac_g),\d^i\ab||_{L^2(S_{\ub,s'})}\\
        &\lesssim \delta a^\frac 12+\int_{-1}^s \mathcal O^2 \delta^2 a |s'|^{-2} ds'+\Big(\int_{-1}^s |s'|^{-4} ds'\Big)^\frac 12 ||s^3 s^i\d^i\aa||_{L^2(\hbub)}\\
        &\lesssim \delta a^\frac 12+\o^2 \delta^2 a |s|^{-1}+\RR[\aa]\, \delta^\frac 52 a^\frac 32 |s|^{-\frac 32}\\
        &\lesssim (1+\o^2 a^{-\frac 12}+\RR[\ab]\, a^{-\frac 12})\, \delta a^\frac 12\lesssim \delta a^\frac 12.
    \end{split}
\end{equation*}
\end{proof}

\begin{proposition}\label{Prop:estimate-of-chih}
    We have $||s^k\d^k\chih||_{L^2(S_{\ub,s})}\lesssim a^\frac 12$, $k\leq N$.
\end{proposition}
\begin{proof}
 We apply  Lemma \ref{Lemma:transport-3} (with $\lambda=\frac 12$)    to the equation  for  $\nab_3\chih $ written schematically   $\nab_3\chih
=-\frac 1 2 \trchb \chih+\Gac_g\cdot\Gac_b$
\begin{equation*}
    \begin{split}
        ||s^{i}\d^i\chih||_{L^2(S_{\ub,s})}&\lesssim ||s^{i}\d^i\chih||_{L^2(S_{\ub,-1})}+\int_{-1}^s |s'|^{i} ||\d^i (\Gac_g\cdot\Gac_b)||_{L^2(S_{\ub,s'})}\, ds'\\
        &\lesssim a^\frac 12 +\int_{-1}^s \o \delta a^\frac 12 |s'|^{-2}\cdot \o \, ds'\\
        &\lesssim (1+\o^2 \delta |s|^{-1})a^\frac 12\lesssim a^\frac 12.
    \end{split}
\end{equation*}
\end{proof}

\begin{proposition}
    We have $||s^k\d^k(\trch,\atrch)||_{L^2(S_{\ub,s})}\lesssim 1+\RR_k[\rho,\rhod]$, $k\leq N$.
\end{proposition}
\begin{proof}
We apply  Lemma \ref{Lemma:transport-3} (with $\lambda=\frac 12$)    to the $\nab_3$  equations   for $\trch$ and $\atrch$
of the form
\beaa
  \nab_3(\trch,\atrch)+\frac 12\trchb (\trch,\atrch)=-\chih\cdot \chibh+2(\rho,-\rhod).
  \eeaa
Making use of the $L^2$ bounds of $\chibh$ and $\chih$ and the corresponding $L^\infty$ bounds, derived using  Lemma \ref{Lemma:Sobolev},  we obtain
\begin{equation}\label{eq:intermediate-estimate-trch}
    \begin{split}
        ||s^{i}\d^i(\trch,\atrch)&||_{L^2(S_{\ub,s})}\lesssim ||s^{i}\d^i(\trch,\atrch)||_{L^2(S_{\ub,-1})}+\int_{-1}^s |s'|^{i} ||\d^i (\chibh\cdot\chih)||_{L^2(S_{\ub,s'})}\\
        &\quad +\int_{-1}^s |s'|^{i} ||\d^i (\rho,\dual\rho)||_{L^2(S_{\ub,s'})}\\
        &\lesssim 1+\sum_{j\leq i/2}\int_{-1}^s a^\frac 12 \cdot ||s^j\d^j\chibh||_{L^\infty(S_{\ub,s'})}+\frac{\delta a^\frac 12}{|s'|}\cdot ||s^j\d^j\chih||_{L^\infty(S_{\ub,s'})} ds'\\
        &\quad+\Big(\int_{-1}^s |s'|^{-2} ds'\Big)^\frac 12 ||s(s^i\d^i(\rho,\rhod))||_{L^2(\hbub)}\\
        &\lesssim  1+\int_{-1}^s \delta a |s'|^{-2} ds'+\RR_i[\rho,\rhod]\, \delta^\frac 12 a^\frac 12 |s|^{-\frac 12}\\
        &\lesssim 1+\RR_i[\rho,\rhod]+\delta a |s|^{-1}\lesssim 1+\RR_i[\rho,\rhod].
    \end{split}
\end{equation}
\end{proof}

\subsection{\texorpdfstring{$L^2$}{L2} Estimate of \texorpdfstring{$f$}{f}}
To derive the estimate of $f$, we make use of the equation, see  \eqref{eq:PTframe-property2},
 \beaa
    \nab_3 f+\frac 12 \trchb f=2\zeta -\chibh\cdot f, \quad f\big|_{H_{-1}}=0.
\eeaa
\begin{proposition}
    We have $||s^k\d^k f||_{L^2(S_{\ub,s})}\lesssim \RR_k[\bb]\, \delta a^\frac 12$.
\end{proposition}
\begin{proof}
 Applying  Lemma \ref{Lemma:transport-3} (with $\lambda=\frac 12$)   and using the $L^2(S)$ estimate of $\zeta$  obtained\footnote{It is essential here  to use  the better  result of Proposition \ref{Prop:estimate-ze} rather than  $\ze\in \Gac_g\sim \delta a^\frac 12 |s|^{-2}$.}   in Proposition \ref{Prop:estimate-ze}, we  derive
\begin{equation*}
    \begin{split}
        ||s^{i}&\d^i f||_{L^2(S_{\ub,s})}\lesssim ||s^{i}\d^i f||_{L^2(S_{\ub,-1})}+\int_{-1}^s |s'|^{i} ||\d^i \zeta,\d^i(\Gac_g\cdot f)||_{L^2(S_{\ub,s'})}\, ds'\\
        &\lesssim 0+\int_{-1}^s \delta a^\frac 12+\RR_i[\bb]\, \delta^\frac 32 a |s'|^{-\frac 32}+\o^2 \delta^2 a |s'|^{-2}\,  ds'\\
        &\lesssim \delta a^\frac 12+\RR_i[\bb]\, \delta^\frac 32 a |s|^{-\frac 12}+\o^2 \delta^2 a |s|^{-1}\\
        &\lesssim \delta a^\frac 12 (1+\RR_i[\bb]\, \delta^\frac 12 a^\frac 12 |s|^{-\frac 12}+\o^2 \delta a^\frac 12 |s|^{-1})\lesssim (1+\RR_i[\b])\, \delta a^\frac 12.
    \end{split}
\end{equation*}
\end{proof}

\subsection{\texorpdfstring{$L^2(S)$}{L2S}-estimates for  the  curvature components}

In what follows we derive non-top $L^2(S_{\ub,s})$  estimates\footnote{That is ignoring loss of derivatives.}   of  the  curvature components. 
%Here, nonlinear terms are also never in the top order.
We start with the following.
\begin{proposition}
    We have $||s^k\d^k\aa||_{L^2(S_{\ub,s})}\lesssim \mathcal R[\aa]\, \delta^\frac 52 a ^\frac 32 |s|^{-\frac 72}$, $k\leq N-1$.
\end{proposition}
\begin{proof}
    Applying Lemma \ref{Lemma:transport-3} to ``$\nab_3\ab=\nab_3\ab$", we have
    \begin{equation*}
    \begin{split}
        ||s^{-1+i}&\d^i\ab||_{L^2(S_{\ub,s})}\lesssim ||s^{-1+i}\d^i\ab||_{L^2(S_{\ub,-1})}+\int_{-1}^s |s'|^{-4} ||s'^3 s'^i\d^i \ab||_{L^2(S_{\ub,s'})}\, ds'\\
        & \lesssim \delta^3 a^2 + \Big(\int_{-1}^s |s'|^{-8} ds'\Big)^\frac 12 ||s^3 s^i\d^i \ab||_{L^2(\hbub)}\\
        &\lesssim \delta^3 a^2+\RR[\ab]\, \delta^\frac 52 a^\frac 32 |s|^{-\frac 72}\lesssim \RR[\ab]\, \delta^\frac 52 a^\frac 32 |s|^{-\frac 72}.
    \end{split}
\end{equation*}
\end{proof}

\begin{proposition}
    We have $||s^k\d^k\bb||_{L^2(S_{\ub,s})}\lesssim \delta^2 a^\frac 32 |s|^{-3}$, $k\leq N-1$.
\end{proposition}
\begin{proof}
   We apply   Lemma \ref{Lemma:transport-3} (with $\lambda=2$, $i\leq N-1$)  to the equation
    $ \nab_3\bb +\div\aa =-2\trchb\,\bb+2\aa\c\ze.  $ Making also use of the estimates   for the Ricci coefficients 
     and $\aa$ already obtained,  we derive
\begin{equation*}
    \begin{split}
        ||s^{3+i}&\d^i\bb||_{L^2(S_{\ub,s})}\lesssim ||s^{3+i}\d^i\bb||_{L^2(S_{\ub,-1})}+\int_{-1}^s |s'|^{3+i} ||\d^{i+1} \ab||_{L^2(S_{\ub,s'})}+|s'|^{3+i} ||\d^i (\ab\cdot\zeta)||_{L^2(S_{\ub,s'})}\\
        &\lesssim \delta^2 a^\frac 32+\Big(\int_{-1}^s |s'|^{-2} ds'\Big)^\frac 12 ||s^3 s^{i+1}\d^{i+1}\ab||_{L^2(\hbub)}+\int_{-1}^s |s'|^{3} \o\RR \delta a^\frac 12 |s'|^{-2}\cdot \delta^\frac 52 a^\frac 32 |s'|^{-\frac 72}\, ds'
        \\
        &\lesssim \delta^2 a^\frac 32+\RR\de^\frac 52 a^\frac 32 |s|^{-\frac 12}+\o\RR[\aa]\, \delta^\frac 72 a^2 |s|^{-\frac 32}\\
        &\lesssim (1+\RR \delta^\frac 12 |s|^{-\frac 12}+\o\RR \delta^\frac 32 a^{\frac 12} |s|^{-\frac 32})\, \delta^2 a^\frac 32\lesssim \delta^2 a^\frac 32.
    \end{split}
\end{equation*}
\end{proof}

\begin{proposition}
    We have $||s^k\d^k(\rho,\dual\rho)||_{L^2(S_{\ub,s})}\lesssim \delta a |s|^{-2}$, $k\leq N-1$.
\end{proposition}
\begin{proof}
     We apply   Lemma \ref{Lemma:transport-3} (with $\lambda=\frac 3 2$, $i\leq N-1$)  
        \beaa
     \nab_3 \rho+\div\bb&=&-\frac 3 2 \trchb \rho +\ze  \c\bb-\frac{1}{2}\chih\c\aa,
 \\
   \nab_3 \rhod+\curl\bb&=&-\frac 3 2 \trchb \rhod+\ze \c\dual \bb-\frac 1 2 \chih\c\dual \aa,
\eeaa
\begin{equation*}
    \begin{split}
        ||s^{2+i}&\d^i\rho||_{L^2(S_{\ub,s})}\lesssim ||s^{2+i}\d^i\rho||_{L^2(S_{\ub,-1})}+\int_{-1}^s |s'|^{2+i} ||\d^{i+1}\bb,\d^i (\zeta\cdot\bb),\d^i(\chih\cdot\ab)||_{L^2(S_{\ub,s'})}\\
        &\lesssim \delta a+(\int_{-1}^s |s'|^{-2} ds')^\frac 12 ||s^2 s^{i+1}\d^{i+1}\bb||_{L^2(\hbub)}+\int_{-1}^s |s|^2 \o\RR \delta a^\frac 12 |s'|^{-2} \cdot\delta^2 a^\frac 32 |s'|^{-3}\\
        &\quad +\int_{-1}^s |s'|^2 \o\RR a^\frac 12 |s'|^{-1} \cdot \delta^\frac 52 a^\frac 32 |s'|^{-\frac 72}\\
        &\lesssim \delta a +\RR[\bb]\delta^\frac 32 a |s|^{-\frac 12}+\o\RR \delta^3 a^2 |s|^{-2}+\o\RR \delta^\frac 52 a^2 |s|^{-\frac 32}\\
        &\lesssim (1+\RR[\bb]a^{-\frac 12}+\o\RR a^{-\frac 12})\delta a\lesssim \delta a. 
    \end{split}
\end{equation*}
The estimate of $\rhod$ follows in the same way.
\end{proof}

%In order to improve $\xi$, we derive the following estimate.
\begin{proposition}
    We have $||s^k\d^k\b||_{L^2(S_{\ub,s})}\lesssim a^\frac 12 |s|^{-1}$, $k\leq N-1$.
\end{proposition}
\begin{proof}
  We apply   Lemma \ref{Lemma:transport-3} (with $\lambda=1$, $i\leq N-1$)   to $  \nab_3\b+\div\rho=-\trchb\b+2\bb\c\chih$
\begin{equation*}
    \begin{split}
        ||s^{1+i}&\d^i\b||_{L^2(S_{\ub,s})}\leq ||s^{1+i}\d^i\b||_{L^2(S_{\ub,-1})}+\int_{-1}^s |s'|^{1+i} ||\d^{i+1}\rho,\d^i (\bb\cdot\chih)||_{L^2(S_{\ub,s'})}\\
        &\lesssim a^\frac 12 +\Big(\int_{-1}^{s} |s'|^{-2}ds'\Big)^\frac 12 ||s(s^{1+i}\d^{1+i}\rho)||_{L^2(\hbub)}+\int_{-1}^s |s|\o\RR a^\frac 12 |s|^{-1}\cdot \delta^2 a^\frac 32 |s'|^{-3}\\
        &\lesssim a^\frac 12+\RR[\rho] \delta^\frac 12 a^\frac 12 |s|^{-\frac 12}+\o\RR \delta^2 a^2 |s|^{-2}\lesssim a^\frac 12.
    \end{split}
\end{equation*}
%实际上还有个wtb \b,这个用decay和flux出来是 \de a^\frac 12 (\int s^{-2} ds)^1/2 \de^{-1/2} a^\frac 12=\de^\1/2 a |s|^{-1/2}，不比1/2好。所以必须利用N-1阶和bootstrap或者Gronwall
\end{proof}

\begin{proposition}
    We have $||s^k\d^k\a||_{L^2(S_{\ub,s})}
    \lesssim \delta^{-1} a^\frac 12$, $k\leq N-1$.
\end{proposition}
\begin{proof}
   We apply\footnote{We use $\rho$  in the  estimates below  to represent both $\rho$ and $\rhod$.  We  omit the estimate for  $\zeta\hot\b$ as it is even better than $\Gac_g\cdot \a$. }   Lemma \ref{Lemma:transport-3} (with $\lambda=1/2$, $i\leq N-1$)   to
    $
        \nab_3\a-  \nab\hot \b=-\frac 1 2 \trchb\a+ \ze\hot \b - 3 (\rho\chih +\rhod\dual\chih)$
    \begin{equation*}
    \begin{split}
        ||s^{i}&\d^i\a||_{L^2(S_{\ub,s})}\lesssim ||s^{i}\d^i\a||_{L^2(S_{\ub,-1})}+\int_{-1}^s |s'|^{i} ||\d^{i+1}\b,\d^i(\rho\cdot\chih)||_{L^2(S_{\ub,s'})}\\
        &\lesssim \delta^{-1} a^\frac 12+\Big(\int_{-1}^s |s'|^{-2} ds'\Big)^\frac 12 ||s^{i+1}\d^{i+1}\b||_{L^2(\hbub)}+\int_{-1}^s \o\RR a^\frac 12 |s|^{-1}\cdot \delta a |s'|^{-2}\\
        &\lesssim \delta^{-1} a^\frac 12+\RR[\b]\, \delta^{-\frac 12} a^\frac 12 |s|^{-\frac 12}+\o\RR \delta a^\frac 32 |s|^{-2}\\
        &\lesssim \delta^{-1} a^\frac 12 (1+\RR[\b]\, \delta^\frac 12 |s|^{-\frac 12}+\o\RR \delta^2 a |s|^{-2})\lesssim \delta^{-1} a^\frac 12.
    \end{split}
\end{equation*}
%We have
%\beaa
%||\a||_{L^2(S)}&\leq & init+\int ||\nab\b||_{L^2(S)}+||\zeta\cdot\b||+||\rho\chih||_{L^2(S)}ds\\
%&\leq & init+(\int |s|^{-2} ds)^\frac 12 ||s\nab\b||_{L^2(\hbub)}+\int ||\chih||_{L^\infty(S)}||\rho||_{L^2(S)} ds\\
%&\leq & init+\mathcal R[\b]|s|^{-\frac 12} \delta^{-\frac 12} a^\frac 12+(\int a |s|^{-4} ds)^\frac 12 ||s\rho||_{L^2(\hbub)}\\
%&\leq & init+\mathcal R[\b]\delta^{-1} +a^\frac 12 |s|^{-\frac 32}\cdot\mathcal R[\rho] \delta^\frac 12 a^\frac 12\\
%&\leq &\delta^{-1} a^\frac 12.
%\eeaa
\end{proof}

\subsection{Improved estimate  for  \texorpdfstring{$\xi$}{xi}}
We improve the estimate   for $\xi$ obtained in Proposition \ref{Prop:estimate-xi} for all but the top derivatives. 
\begin{proposition}
    We have $||s^k\d^k \xi||_{L^2(S_{\ub,s})}\lesssim a^\frac 12$, $k\leq N-1$.
\end{proposition}
\begin{proof}
We proceed as in the proof of  Proposition \ref{Prop:estimate-xi}, starting with  $\nab_3\xi=\chih\cdot \zeta+\Gac_b\cdot \zeta+\b$,     by taking into account the  new bound for $\b$ just derived.
Thus, for $i\leq N-1$,
\begin{equation*}
    \begin{split}
    ||\, |s|^{i-1}\d^i\xi ||_{L^2(S_{\ub,s})}&\lesssim  ||\, |s|^{i-1}\d^i\xi||_{L^2(S_{\ub,-1})}+\int_{-1}^{s} |s'|^{i-1} ||\d^i(\chih\cdot\zeta)||_{L^2(S_{\ub,s'})}ds'\\
    &\quad +\int_{-1}^{s}|s'|^{i-1} ||\d^i\b||_{L^2(S_{\ub,s'})} ds' \\
    &\lesssim  0+\int_{-1}^{s} |s'|^{-1}\o a^\frac 12 \cdot \delta a^\frac 12 |s'|^{-2}+\int_{-1}^s a^\frac 12 |s'|^{-2}\, ds'\\
    &\lesssim  \mathcal O^2 \delta a |s|^{-2}+a^\frac 12 |s|^{-1}=  a^\frac 12 |s|^{-1} (1+\o^2 \delta a^{\frac 12} |s|^{-1})\lesssim a^\frac 12 |s|^{-1}.
    \end{split}
\end{equation*}
\end{proof}

\subsection{Summary of the results proved  in this  section}

\begin{proposition}
\lab{Prop:summary-Ricci}
The following  estimates hold true, see  \eqref{eq:compundnorms} for the definition of the norms,
  \beaa
   \o_{\le N} + \RR^S_{\le N-1}  \les \RR_{\le N}.
  \eeaa
 We also have the improved estimates
 \begin{equation*}
     \o_{\leq N}[\chih,\wtb,\chibh,\zeta,\xi]+\RR^S_{\le N-1}[\bb,\rho,\rhod,\b,\a]\lesssim 1.
 \end{equation*}
\end{proposition}

With the  help of the  non-integrable Sobolev estimates of Lemma \ref{Lemma:Sobolev} we also obtain $O_{\leq N-3,\infty}\lesssim \RR_{\leq N}$. We state the precise estimates:
\begin{corollary}
For $k\leq N-3$ we have  the estimates:
\begin{equation*}
    ||s^k\d^k \chih||_{L^\infty(S_{\ub,s})}\leq \frac{a^\frac 12}{|s|} ,\quad ||s^k\d^k\om||_{L^\infty(S_{\ub,s})}\leq \RR[\rho]\frac{\delta^\frac 12 a^\frac 12}{|s|^\frac 32},\quad ||s^k\d^k \trch||_{L^\infty(S_{\ub,s})}\lesssim \RR[\rho]\frac{1}{|s|},
\end{equation*}
\begin{equation*}
    %||s^k\nab^k \xi||_{L^\infty(S_{\ub,s})}\leq \RR[\b]\frac{\delta^{-\frac 12}a^\frac 12}{|s|^\frac 12},\quad 
||s^k\d^k(\widecheck{\trchb},\chibh,\zeta)||_{L^\infty(S_{\ub,s})}\leq \frac{\delta a^\frac 12}{|s|^2}, \quad ||s^k\d^k\xi||_{L^\infty(S_{\ub,s})}\lesssim \frac{a^\frac 12}{|s|},\quad ||s^k \d^k f||_{L^\infty(S_{\ub,s})}\lesssim \mathcal R[\bb] \frac{\delta a^\frac 12}{|s|},
\end{equation*}
%Whether there is a bootstrap constant depends on the $L^2(S)$ estimate in last subsection.
\begin{equation*}
    ||s^k\d^k\ab||_{L^\infty(S_{\ub,s})}\lesssim \mathcal R[\ab]\frac{\delta^\frac 52 a^\frac 32}{|s|^\frac 92},\quad ||s^k\d^k\bb||_{L^\infty(S_{\ub,s})} \lesssim \frac{\delta^2 a^\frac 32}{|s|^4},\quad ||s^k\d^k(\rho,\dual\rho)||_{L^\infty(S_{\ub,s})}\lesssim \frac{\delta a}{|s|^3},
\end{equation*}
\begin{equation*}
    ||s^k\d^k \b||_{L^\infty(S_{\ub,s})}\lesssim \frac{a^\frac 12}{|s|^2},\quad ||s^k\d^k\a||_{L^\infty(S_{\ub,s})}\lesssim \frac{\delta^{-1}a^\frac 12}{|s|}.
\end{equation*}
\end{corollary}

%%%%%%%%%%%%%%%%%

\section{Energy Estimates}\label{Section:energy-estimates}
{\bf  Notations.}  We make the following notational  conventions    to be use throughout this section.
\begin{enumerate}
\item     Whenever we use the index $i_1$, we mean  summation  over $i_1=0,1,\cdots,i$ (with the covention that replaces $i_1-1$ by $0$ if $i_1=0$); whenever we use the index $i_2$, we mean  summation  over $i_2=0,1,\cdots,[i/2]$. In these situations we drop the summation  symbol.

\item We  use $S$,  when there is no ambiguity,  to denote the spheres $S_{\ub,s}$.
\item We  use the  double integral sign $\iint$  to  denote  either  a full  spacetime integral or the integral over the $\ub$, $s$ variable (i.e. the non-angular variables).
\end{enumerate}

%%%%%%%%%%%

\subsection{Integrating region}\label{subsect:Integrating region}

%%%%%%%%%%%%%
Recall $\tau=\frac 1{10}a\ub+s$.  In the region $\{\tau\leq \tau^*\}$ with $\tau^*\leq -\frac 18 a\delta$, we have 
\bea
    \eg_3(\tau)=\eg_3(s)=1,\quad \eg_4(\tau)=\frac 1{10}a+\eg_4(s).
\eea
We need to estimate $\eg_4(s)$ for which we  use the formula,  see \eqref{eq:intro-geodesic}, $-2\omg=\eg_3(\eg_4(s))$. Moreover, the estimates in the PT frame, together with the transformation formula for $\om$ (see Lemma \ref{Proposition:transformationRicci})  imply that
\begin{equation*}
    |\omg|\lesssim |\om|+|f|\cdot|(\zeta,\etab)|+|s|^{-1}|f|^2\lesssim \o\delta^\frac 12 a^\frac 12 |s|^{-\frac 32}+\o^2 \delta^2 a |s|^{-3}.
\end{equation*}
Then using $e_4(s)=0$ on $H_{-1}$ and integrating in $\eg_3=\pa_s$ direction, we obtain
\begin{equation}\label{eq:eS_4(s)}
    |\eg_4(s)|\lesssim \o\delta^\frac 12 a^\frac 12 |s|^{-\frac 12}+\o^2 \delta^2 a |s|^{-2}\lesssim \o\ll a/10.
\end{equation}
In particular $\eg_4(\tau)>0$. 
Let  $(\grad \tau)^\mu $  be the vectorfield   perpendicular to the level surfaces of $\tau$ defined by $(\grad \tau)^\mu =g^{\mu\nu} \pr_\nu \tau$. We have
\beaa
(\grad \tau)^\mu =g^{\mu\nu} \pr_\nu \tau= -\frac 12 \big(\eg_3(\tau)  \eg_4+ \eg_4(\tau)\eg_3\big)=-\frac 12\big(  \eg_4+  \eg_4(\tau)\eg_3\big),
\eeaa
so
\beaa
 g(\grad\tau,\grad\tau)&=&- 2\eg_4(\tau)=-2\big(a+\eg_4(s)\big)
\eeaa
which is strictly  negative. This shows, in  particular, that    $\Sigma_\tau$ is a spacelike hypersurface with future unit normal  given by
\beaa
N_\tau=-\frac{\grad \tau}{|\grad\tau|}.
\eeaa

%%%%%%%%%%

\subsection{Divergence Lemma}

%%%%%%%%%%%%%%%%

We apply the spacetime divergence Lemma (see e.g. \cite{CK}, \cite{GKS}) to  causal domains of the form
$\MM\subset  \MM(\delta, a;\tau_*)$ enclosed by $\Si_\tau =\{\tau=\mathrm{const}\}\cup \Hb_{\ub} $ to the future and  $H_{-1}\cup \underline H_0$ to the past. 
\begin{lemma}
Consider  a vectorfield $ X$   on a causal  domain  $\MM\subset  \MM(\delta, a;\tau_*)$ enclosed by $\Si_\tau =\{\tau=\mathrm{const}\}\cup \Hb_{\ub} $ to the future and  $H_{-1}\cup \underline H_0$ to the past.  Then
\beaa
\int_{\Si_\tau} g(X, N_\tau) +\int_{\Hb_{\ub}} g(X, e_3) =\int_{H_{-1}} g(X, e_4)+\int_{\Hb_0 } g(X, e_3) - \int_\MM ( \mathrm{Div} X)  
\eeaa
where  the integrations on $\Si_\tau$  and $\MM$ are with respect to their  standard area and volume forms and   $N_\tau=-\frac{\mathrm{grad}\, \tau}{|\mathrm{\\grad}\, \tau|}$ is  the  future unit normal to $\Si_{\tau}$.   The integrations on  the null hypersurfaces  $\hbub$ and $H_{-1}$ of scalar functions $f$  are defined as follows
\beaa
\int_{\Hb_{\ub}} f =\int_s ds \int_{S_{\ub, s}}  f d\vol_S, \qquad  \int_{H_{-1}} f =\int_{\ub}  d\ub  \int_{S_{\ub,-1}}   f d\vol_S
\eeaa
%\Red{ Should  there be  a factor of $1/2$ ?} \Blue{I think the formula here is correct: For volume element we have $2ds\wedge d\ub\cdots$ and in the mean time $e_3=-2\grad \ub$, i.e., $(e_3)_\mu=-2(d\ub)_\mu$.}
\end{lemma}
\begin{proof}
Immediate application of the Stokes formula  applied to the  differential form $(\dual X)_{\a\b\ga}=\in_{\a\b\ga\mu} X^\mu$ by observing that
$\dual ( d\dual X)=\Div X$. See  Section 8.1 in   \cite{CK} for the details.
\end{proof}
\begin{corollary}
\lab{corrllary:diverg.Thm}
Consider the vectorfield $X=\la_3 e_3 +\la_4 e_4$,  where $\la_3,\la_4 $ are given   smooth functions. Then, integrating on the same domain $\MM$,
\bea
\bsplit
& \int_{\Si_\tau} \frac{1}{|\mathrm{grad}\, \tau| }\big(  \la_3 +\la_4 (a+ e_4(s) \big)+\int_{\Hb_{\ub} } 2 \la_4 &=&\int_{\Hb_{0} } 2 \la_4 +\int_{H_{-1}}  2 \la_3 + \int_\MM ( \mathrm{\Div} X) 
\end{split}
\eea
with $|\mathrm{grad} \,\tau|\approx a^{1/2}$.
%\Blue{(Check in Minkowski: For the energy momentum tensor $Q_{\mu\nu}$ with multiplier $\pa_t$, $P^\mu=Q_{0}^{\; \mu}$, so $P=-\frac 12 Q_{0e_3}e_4-\frac 12 Q_{0e_4}e_3$, so $\lambda_3$ and $\lambda_4$ correspond to $\frac 12 Q(\pa_t,e_4)$ and $\frac 12 Q(\pa_t,e_3)$.)}
\end{corollary}
\begin{proof}
We have
 \beaa
 g(X, N_\tau)&=&- \frac{1}{|\grad \, \tau| }\big( \la_3 e_3(\tau)+\la_4 e_4(\tau) \big)=- \frac{1}{|\grad \,\tau| }\big(  \la_3 +\la_4 (a+ e_4(s)) \big)\\
 g(X, e_3)&=&- 2 \la_4, \qquad  g(X, e_4)=- 2 \la_3 
 \eeaa
and the result follows from the lemma.

\end{proof}

\subsection{Estimates for   general Bianchi pairs}\label{subsect:general-Bianchi-pair}

\def\ss{\mathfrak{s}}
\begin{definition}
    We denote $\ss_0$ by the set of pairs of scalar fields in the spacetime, $\ss_1$ by the set of $\HH$-horizontal $1$-forms, and $\ss_2$ by the set of symmetric traceless $\HH$-horizontal covariant $2$-tensors.
\end{definition}

\def\DD{\slashed{\mathcal{D}}}
\begin{definition}
    We consider  the non-integrable horizontal Hodge-type operators\footnote{ See  \cite{CK} for the original definitions      and  \cite{GKS} for the  extensions to the non-integrable case.}
    \begin{itemize}
        \item $\DD_1$ takes $\ss_1$ into $\ss_0$:\qquad     $\DD_1\xi=(\div\xi,\curl\xi),$
        \item $\DD_2$ takes $\ss_2$ into $\ss_1$:\qquad $   (\DD_2\xi)_a=\nab^b \xi_{ab},$
        \item $\DD_1^*$ takes $\ss_0$ into $\ss_1$:\qquad $    (\DD_1^* (f,f_*))_a=-\nab_a f+\in_{ab} \nab_b f_*,$
        \item $\DD_2^*$ takes $\ss_1$ into $\ss_2$:\qquad $  \DD_2^* \xi=-\frac 12\nab\hot\xi$.
    \end{itemize}
\end{definition}
\begin{lemma}\label{Lemma:Hodge-Leibniz}
    The following identities hold:
    \bea
    \bsplit
        \DD_1^*(f,f_*)\cdot u&=(f,f_*)\cdot\DD_1 u-\nab_a (f u^a+f_*(\dual u)^a),\quad (f,f_*)\in \ss_0,\quad u\in \ss_1,
    \\
        (\DD_2^* f)\cdot u&=f\cdot (\DD_2 u)-\nab_a(f_b u^{ab}),\quad f\in \ss_1,\quad u\in\ss_2.
        \end{split}
    \eea
\end{lemma}
\begin{proof}
Direct calculation. See Lemma 2.1.23 in \cite{GKS}.
\end{proof}
\begin{definition}
We consider the following  two types of  abstract Bianchi pairs\footnote{ The  null Bianchi identities  in Proposition \ref{Prop:transport.e_3}  can be split in  the pairs $(\a, \b)$, $\big(\b, (\rho, \rhod)\big)$, $\big((\rho, \rhod), \bb\big)  $ and $(\bb, \aa)$ which  fit into one of the two types described here.}:
\begin{enumerate}
\item[Type I.]   These are systems in  $\psi_1\in \ss_k$ and $\psi_2\in \ss_{k-1}$ ($k=1,2$) of the form
\begin{equation}\label{eq:null-Bianchi-model-pair}
    \begin{split}
        \nab_3 \psi_1+\lambda\trchb \psi_1&=-k \DD_k^*\psi_2+F_1,\\
        \nab_4 \psi_2&=\DD_k\psi_1+F_2,
    \end{split}
\end{equation}
\item[Type II.]     These are systems in  $\psi_1\in \ss_{k-1}$ and $\psi_2\in \ss_k$ ($k=1,2$) of the form
\begin{equation}\label{eq:null-Bianchi-model-pair2}
    \begin{split}
        \nab_3 \psi_1+\lambda\trchb \psi_1&=\DD_k\psi_2+F _1,\\
        \nab_4 \psi_2&=-k\DD_k^*\psi_1+ F_2,
    \end{split}
\end{equation}

\end{enumerate}
\end{definition}

The main goal  of  this subsection is to prove the following lemma:
\begin{lemma}
\lab{Lemma:Main-Bianchipairs}
    Suppose that the bootstrap assumption holds. Then, for both pairs \eqref{eq:null-Bianchi-model-pair}, \eqref{eq:null-Bianchi-model-pair2}, we have the estimate
    \bea
    \label{eq:Main-Bianchipairs}
    \begin{split}
        & a^{-\frac 12} \int_{\Sigma_\tau} s^{2i+4\lambda-2}|\d^i\psi_1|^2+\int_{\hbub} s^{2i+4\lambda-2}|\d^i\psi_2|^2\lesssim \int_{H_{-1}}s^{2i+4\lambda-2}|\d^i\psi_1|^2 +\int_{\Hb_0} s^{2i+4\lambda-2}|\d^i\psi_2|^2\\
        &\, +\sum_{j=0}^i |s|^{-j} \Big|\iint_{\MM} s^{2i+4\lambda-2}\d^i\psi_1\cdot \d^{i-j}F_1\Big|+\Big|\iint_{\MM} s^{2i+4\lambda-2}\d^i\psi_2\cdot \d^i  F_2\Big|.
    \end{split}
\eea
\end{lemma}

\begin{proof}
We prove the estimate for the first type; the second type follows in the same way. Commuting the equation with $\d^i$,  using  the commutator   Lemma\footnote{We deal with the commutation between $\nab_3$ and $\d^i$ in the same way as in Section \ref{Section:Integrating-Ricci}. } \ref{le:comm-PTframe1},  we 
derive\footnote{ Here the Hodge-type operators act with the indices of $\psi_1$, $\psi_2$, e.g., 
  $  (\DD_2\d^i\a)_a=\nab^b \d^i\a_{ab},\quad (\DD_2^*\d^i\b)_{ab}=-\frac 12 \nab_a\hot(\d^i\b)_{b}.$} 
\bea
    \begin{split}
        \nab_3 \d^i\psi_1+\Big(\lambda+\frac i2\Big)\, \trchb\,  \d^i\psi_1&=-k \DD_k^*\d^i\psi_2+F _1^i,\\
        \nab_4 \d^i\psi_2&=\DD_k\d^i\psi_1+ F_2^i,
        \end{split}
        \eea
        where $F_1^0=F_1, F_2^0 =F_2$,  as in \eqref{eq:null-Bianchi-model-pair}-\eqref{eq:null-Bianchi-model-pair2}, and  for $i\ge 1$
        \bea
        \lab{eq:H1H2}
        \bsplit
        F_1^i&=\sum_{j=0}^i |s|^{-j}\d^{i-j}F_1+\d^{i-1}(\Gac_g\cdot  \d\psi_1)+\d^{i-1}(\bb^*\cdot\psi_1)
        +k[\DD^*_k,\d^i]\psi_2,\\
          F_2^i&=\d^i  F_2+[\nab_4,\d^i]\psi_2-[\DD_k,\d^i]\psi_1.
    \end{split}
\eea
We next  make use of  the   formulas
\bea
\Div e_3 =-2\omb+\trchb=\trchb, \qquad \Div e_4=-2\om+\trch
\eea
 to  calculate the divergence of the vectorfield  
\beaa
X=s^{2i+4\lambda-2}|\d^i\psi_1|^2 e_3+k s^{2i+4\lambda-2}|\d^i\psi_2|^2e_4,
\eeaa
\begin{equation*}
\begin{split}
    \Div X
    &=s^{2i+4\lambda-2} |\d^i\psi_1|^2\Div e_3+e_3(s^{2i+4\lambda-2}|\d^i \psi|^2)+ks^{2i+4\lambda-2} |\d^i\psi_2|^2\Div e_4\\
    &\quad +ke_4(s^{2i+4\lambda-2}|\d^i \psi_2|^2)\\
    &=s^{2i+4\lambda-2}|\d^i\psi_1|^2\trchb+(2i+4\lambda-2) s^{2i+4\lambda-3}|\d^i\psi_1|^2+2s^{2i+4\lambda-2}\d^i\psi_1\cdot \nab_3\d^i\psi_1\\
    &\quad +ks^{2i+4\lambda-2} |\d^i\psi_2|^2(-2\om+\trch)+k(2i+4\lambda-2) s^{2i+4\lambda-3}e_4(s)|\d^i\psi_2|^2\\
    &\quad +2ks^{2i+4\lambda-2}\d^i\psi_2\cdot \nab_4\d^i\psi_2\\
    &=s^{2i+4\lambda-2}|\d^i\psi_1|^2\trchb+(2i+4\lambda-2)s^{2i+4\lambda-3}|\d^i\psi_1|^2\\
    &\quad +2s^{2i+4\lambda-2}\d^i\psi_1\cdot \left(-k\DD_k^*\psi_2-\Big(\frac i2+\lambda\Big)\trchb\d^i\psi_1+F _1^i\right)\\
    &\quad +ks^{2i+4\lambda-2} |\d^i\psi_2|^2(-2\om+\trch)+k(2i+4\lambda-2) s^{2i+4\lambda-3}e_4(s)|\d^i\psi_2|^2\\
    &\quad +2ks^{2i+4\lambda-2}\d^i\psi_2\cdot (\DD_k\d^i\psi_1+ F_2^i)\\
    &=(i+2\lambda-1)s^{2i+4\lambda-2}\Big(\frac{2}{s}-\trchb\Big)|\d^i\psi_1|^2+2s^{2i+4\lambda-2}\d^i\psi_1\cdot F_1^i-2ks^{2i+4\lambda-2} \d^i\psi_1\cdot\DD_k^*\d^i\psi_2\\
    &\quad +ks^{2i+4\lambda-2}\left(-2\om+\trch+(2i+4\lambda-2)s^{-1}e_4(s)\right)|\d^i\psi_2|^2+2ks^{2i+4\lambda-2}\d^i\psi_2\cdot  F_2^i\\
    &\quad +2ks^{2i+4\lambda-2}\d^i\psi_2\cdot\DD_k\d^i\psi_1.
    \end{split}
\end{equation*}
Note that
\begin{equation*}
    \d^i\psi_2\cdot\DD_k\d^i\psi_1- \d^i\psi_1\cdot\DD_k^*\d^i\psi_2=\div (\d^i\psi_1\cdot \d^i \psi_2),
\end{equation*}
which is a direct generalization of Lemma \ref{Lemma:Hodge-Leibniz}. Therefore we get
\begin{equation*}
    \begin{split}
        \Div X
        &=-(i+2\lambda-1)s^{2i+4\lambda-2}\wtb\, |\d^i\psi_1|^2+2s^{2i+4\lambda-2}\d^i\psi_1\cdot F_1^i\\
    &\quad +ks^{2i+4\lambda-2}\left(-2\om+\trch+(2i+4\lambda-2)s^{-1}e_4(s)\right)|\d^i\psi_2|^2+2ks^{2i+4\lambda-2}\d^i\psi_2\cdot  F_2^i\\
    &\quad +2ks^{2i+4\lambda-2}%(\d^i\psi_2\cdot\DD_k\d^i\psi_1- \d^i\psi_1\cdot\DD_k^*\d^i\psi_2)
    \div(\d^i\psi_1\cdot\d^i\psi_2).
    \end{split}
\end{equation*}
The last term equals
\begin{equation*}
\begin{split}
    2k s^{2i+4\lambda-2}
    \div(\d^i\psi_1\cdot\d^i\psi_2)= 2k\div(s^{2i+4\lambda-2}\d^i\psi_1\cdot\d^i\psi_2)-2k(2i+4\lambda-2)s^{2i+4\lambda-2}s^{-1}\nab_a(s) (\d^i\psi_1\cdot\d^i\psi_2)_a.
    \end{split}
\end{equation*}
Therefore,
\bea
\lab{eq:DivX}
\bsplit
\Div X&= 2k\div(s^{2i+4\lambda-2}\d^i\psi_1\cdot\d^i\psi_2) -(i+2\lambda-1)s^{2i+4\lambda-2}\wtb\, |\d^i\psi_1|^2\\
&\quad +ks^{2i+4\lambda-2} \Big(-2\om+\trch+(2i+4\lambda-2)s^{-1}e_4(s)\Big)|\d^i\psi_2|^2\\
&\quad -2k(2i+4\lambda-2)s^{2i+4\lambda-2}s^{-1}\nab_a(s) (\d^i\psi_1\cdot\d^i\psi_2)_a\\
&\quad +2s^{2i+4\lambda-2}\d^i\psi_1\cdot F_1^i+2ks^{2i+4\lambda-2}\d^i\psi_2\cdot  F_2^i.
\end{split}
\eea
To derive our  final result it remains to integrate   \eqref{eq:DivX}  on  $\MM$. 
In view of  the lack of   integrability of our PT frame  we need  however  to  replace $\div$ with $\Div$  with the help of the   formula, for an arbitrary  $\HH$-horizontal $1$-form  $\Psi$,  see\footnote{Note that in  our case, we have $\eta=0$.}  \cite[Lemma 2.1.40]{GKS}
\beaa
       \Div \Psi=\div  \Psi+\frac 12\etab \cdot \Psi.
    \eeaa
 We then  apply Corollary \ref{corrllary:diverg.Thm} to \eqref{eq:DivX}   to derive 
\begin{equation*}
    \begin{split}
        & a^{-\frac 12} \int_{\Sigma_\tau} s^{2i+4\lambda-2}|\d^i\psi_1|^2+\left(a+e_4(s)\right)ks^{2i+4\lambda-2}|\d^i\psi_2|^2+\int_{\hbub} s^{2i+4\lambda-2}|\d^i\psi_2|^2\\
        &\lesssim \int_{H_{-1}}s^{2i+4\lambda-2}|\d^i\psi_1|^2+\int_{\Hb_0} s^{2i+4\lambda-2}|\d^i\psi_2|^2+a^{-\frac 12}\int_{\Sigma_\tau} 2ks^{2i+4\lambda-2} \Big |-\frac 12 f\cdot \d^i\psi_1\cdot\d^i\psi_2\Big |\\
        &\quad +\iint_{\MM} s^{2i+4\lambda-2}|(\Gac_b+s^{-1}e_4(s))| |\d^i\psi_2|^2+|\Gac_g| |\d^i\psi_1|^2+|\etab||\d^i \psi_1||\d^i\psi_2|\\
        &\quad +\Big|\iint_{\MM} s^{2i+4\lambda-2}(\d^i\psi_1\cdot F_1^i+\d^i\psi_2\cdot  F_2^i)\Big|.
    \end{split}
\end{equation*}
Note that on $\Sigma_\tau$, $|s|=|\tau|+\frac 1{10}|a\ub|\geq |\tau|$. Also, the bound of $\eg_4(s)$ \eqref{eq:eS_4(s)} implies $|e_4(s)|\lesssim \mathcal O |s|^{-1}$. Then,
\begin{equation*}
\begin{split}
    \iint_{\MM} s^{2i+4\lambda-2}|\Gac_g||\d^i\psi_1|^2&\, d\vol\lesssim\int_{-1}^{-a\delta} \left(\int_{\Sigma_\tau} \mathcal O\frac{\delta a^\frac 12}{|s|^2}s^{2i+4\lambda-2}|\d^i\psi_1|^2\, a^{-\frac 12}\, d\Sigma_\tau\right)d\tau\\
    & \lesssim a^{-\frac 12}\sup_\tau \Big(\int_{\Sigma_\tau} s^{2i+4\lambda-2}|\d^i\psi_1|^2\, d\Sigma_\tau\Big)\cdot \int_{-1}^{-a\delta} \mathcal{O}\frac{\delta a^\frac 12}{|\tau|^2}  d\tau\\
    &\lesssim a^{-\frac 12}\sup_\tau\Big(\int_{\Sigma_\tau} s^{2i+4\lambda-2}|\d^i\psi_1|^2\, d\Sigma_\tau\Big)\cdot \frac{\mathcal O}{a^\frac 12},
    \end{split}
\end{equation*}
and
\bea
\lab{eq:expression-Bianchi}
    \begin{split}
        \iint_{\MM} |(\Gac_b+s^{-1}e_4(s))||\d^i\psi_2|^2\, d\vol\leq \int_0^\delta \Big(\int_{\hbub} \mathcal{O}|s|^{-1} |\d^i\psi_2|^2\Big)\, d\ub\leq \frac{\mathcal O}{a}\cdot \sup_{\ub} \int_{\hbub} |\d^i\psi_2|^2.
    \end{split}
\eea
Therefore, taking the supremum over $\tau$ and $\ub$, we can absorb several terms on the right by the left hand side and obtain
\beaa
     &&    a^{-\frac 12} \sup_\tau\int_{\Sigma_\tau}  s^{2i+4\lambda-2}|\d^i\psi_1|^2+\sup_{\hbub}\int_{\hbub} s^{2i+4\lambda-2}|\d^i\psi_2|^2\lesssim \int_{H_{-1}}s^{2i+4\lambda-2}|\d^i\psi_1|^2\\
        &&+\int_{\Hb_0} s^{2i+4\lambda-2}|\d^i\psi_2|^2+\Big|\iint_{\MM} s^{2i+4\lambda-2}(\d^i\psi_1\cdot F_1^i+\d^i\psi_2\cdot  F_2^i)\Big|.
\eeaa
It remains to estimate the terms  $F _1,  F_2$ in  \eqref{eq:H1H2}.
We have, schematically, using \eqref{comm:PTframe2} and $\atrchb=0$,
\begin{equation*}
    [\DD,\nab^i]\psi=\sum_{j=0}^{i-1} \nab^j (\, ^{(h)}K\psi+\atrch \nab_3 \psi)
\end{equation*}
where $\DD$ stands for any of the four Hodge-type operators. We also have, by \eqref{comm:PTframe1},
\begin{equation*}
    [\DD,\nab_3^i]\psi=\d^{i-1}(|s|^{-1}\d\psi)+\d^{i-1}(\Gac_g\cdot \d\psi)+\d^{i-1}(\bb\cdot \psi),
\end{equation*}
so composing these two formulas we get
\begin{equation*}
    [\DD,\d^i]\psi=\d^{i-1}\left(|s|^{-1}\d\psi\right)+\d^{i-1}\left(\atrch\cdot \d\psi\right)+\d^{i-1}\left(\bb\cdot\psi+\, ^{(h)}K\cdot\psi\right)
\end{equation*}
%Also, note that in our case $\atrchb=0$ so the $\nab_4\psi$ is in fact absent.
Also,
\begin{equation*}
    [\nab_4,\d^i]\psi=\d^{i-1}\left((\Gac_b,\xi) \cdot \d\psi\right)+\d^{i-1}\left((|s|^{-1}\xi,\dual\b,\rhod)\cdot \psi\right).
\end{equation*}
Therefore
\bea
\bsplit
 F_1^i&=|s|^{-j}\d^{i-j}F_1+ \err_1^i,\\
 F_2^i&=\d^i  F_2+\err_2^i,\\
\\
\err_1^i&=\d^{i-1}\left(\Gac_g\cdot  \d\psi_1\right)+\d^{i-1}\left(\bb^*\cdot\psi_1\right)\\
&+\d^{i}\left(|s|^{-1}\d\psi_2\right)+\d^{i}\left(\atrch\cdot \d\psi_2\right)+\d^{i-1}\left(\bb\cdot\psi_2+\, ^{(h)}K\cdot\psi_2\right),\\
\err_2^i&=\d^{i-1}\left((\Gac_b,\xi)\cdot \d\psi_2\right)+\d^{i-1}\left((|s|^{-1}\xi,\dual\b,\rhod)\cdot \psi_2\right)\\
&  +\d^{i}\left(|s|^{-1}\d\psi_1\right)+\d^{i}\left(\atrch\cdot \d\psi_1\right)+\d^{i-1}\left(\bb\cdot\psi_1+\, ^{(h)}K\cdot\psi_1\right).
\end{split}
\eea
We  deal below  with  the  contributions of         $\err_1, \err_2$ to \eqref{eq:expression-Bianchi}.   
We first deal with the second line in the expression of $\err_1^i$. The estimate of the second line of $\err_2^i$  is  identical.
\begin{equation*}
\begin{split}
    \Big|&\iint_{\MM} |s|^{2i+4\lambda-2}\d^i\psi_1\cdot \left(\d^{i}\big((|s|^{-1},\atrch,\Gac_g)\psi_2\big)+\d^{i-1}\big(\bb\cdot \psi_2+\, ^{(h)}K\cdot \psi_2\big)\right)\Big|\\
    &\lesssim \iint_{\MM} |s|^{2i+4\lambda-2} |\d^i\psi_1|\Big(|s|^{-1}|\d^i\psi_2|+|s|^{-2}|\d^{i-1}\psi_2|\Big)
    \\
    &+\iint_{\MM} \Big(1+|s|\cdot ||s^{i-1}\d^{i-1}\Kh||_{L^2(S)}\Big) \cdot ||s^{i_2+2\lambda-1}\d^{i_2}\psi_2||_{L^\infty(S)}||s^{i_1+2\lambda-1}\d^{i_1} \psi_1||_{L^2(S)}\\
    &\lesssim \Big(\iint_{\MM} |s|^{-2} |s^{i_1+2\lambda-1}\d^{i_1}\psi_1|^2\Big)^\frac 12 \cdot \Big(\iint_{\MM} |s^{i_1+2\lambda-1}\d^{i_1}\psi_2|^2\Big)^\frac 12\\
    &\quad +\RR[\rho]\iint_{\MM} |s|^{-1}  ||s^{i_1+2\lambda-1}\d^{i_1}\psi_2||_{L^2(S)}||s^{i_1+2\lambda-1}\d^{i_1} \psi_1||_{L^2(S)} \\
    &\lesssim C_b\, \Big(\iint_{\MM} |s|^{-2} |s^{i_1+2\lambda-1}\d^{i_1}\psi_1|^2\Big)^\frac 12 \cdot \Big(\iint_{\MM} |s^{i_1+2\lambda-1}\d^{i_1}\psi_2|^2\Big)^\frac 12\\
    &\lesssim C_b\, \left(\int_{-1}^{-a\delta} |\tau|^{-2} \sup_\tau \left(a^{-\frac 12}\int_{\Sigma_\tau} |s|^{2i_1+4\lambda-2}|\d^{i_1}\psi_1|^2\right) d\tau\right)^\frac 12 \cdot \Big(\delta\cdot ||s^{i_1+2\lambda-1}\d^{i_1}\psi_2||_{L^2(\hbub)}^2\Big)^\frac 12\\
    &\lesssim C_b\, (a\delta)^{-\frac 12}\cdot \sup_\tau \left(a^{-\frac 12}\int_{\Sigma_\tau} |s|^{2i_1+4\lambda-2}|\d^{i_1}\psi_1|^2\right)^\frac 12 \cdot \delta^\frac 12 \cdot ||s^{i_1+2\lambda-1}\d^{i_1}\psi_2||_{L^2(\hbub)}\\
    &\lesssim C_b\, a^{-\frac 12} \left(\sup_\tau a^{-\frac 12}\int_{\Sigma_\tau} |s|^{2i_1+4\lambda-2}|\d^{i_1}\psi_1|^2+||s^{i_1+2\lambda-1}\d^{i_1}\psi_2||_{L^2(\hbub)}^2\right),
    \end{split}
\end{equation*}
which,  after summing up over  the index $i$,  can be absorbed by the left hand side of the  desired  estimate \eqref{eq:Main-Bianchipairs}. 

For the first line in the expression of $\err_1^i$, we have
\begin{equation*}
    \begin{split}
        \iint_{\MM} &|s|^{2i+4\lambda-2}\d^i\psi_1\cdot \left(\d^{i-1}(\Gac_g \cdot \d\psi_1+\bb\cdot\psi_1)\right)\\
        &\lesssim \iint_{\MM} \left(\o \delta a^\frac 12 |s|^{-2} +\delta^2 a^\frac 32 |s|^{-4}\cdot |s|\right) ||s^{i_1+2\lambda-1}\d^{i_1}\psi_1||_{L^2(S)}\cdot ||s^{i_1+2\lambda-1}\d^{i_1}\psi_1||_{L^2(S)}\\
        %&\quad +\iint \o \delta a^\frac 12 |s|^{-1} +\delta^2 a^\frac 32 |s|^{-3}\cdot |s|) |s|^{-1}\sum_{j\leq i/2+2}||s^{j+2\lambda-1}\psi_1||_{L^2(S)}\cdot ||s^{i+2\lambda-1}\psi_1||_{L^2(S)}\\
        &\lesssim \int_{-1}^{-a\delta} \left(\o \delta a^\frac 12 |\tau|^{-2} +\delta^2 a^\frac 32 |\tau|^{-3}\right) \sup_\tau \left(a^{-\frac 12}\int_{\Sigma_\tau} |s|^{2i+4\lambda-2}|\d^i\psi_1|^2 d\vol_{\Sigma_\tau}\right) d\tau\\
        &\lesssim (\o a^{-\frac 12}+a^{-\frac 12})\sup_\tau \left(a^{-\frac 12}\int_{\Sigma_\tau} |s|^{2i+4\lambda-2}|\d^i\psi_1|^2 d\vol_{\Sigma_\tau}\right).
    \end{split}
\end{equation*}
Similarly, for the first line in the expression of $\err_2^i$,  we derive
\begin{equation*}
    \begin{split}
\iint_{\MM}& |s|^{2i+4\lambda-2}\d^i\psi_2\cdot \left(\d^{i-1}\big((\Gac_b,\xi) \cdot \d\psi_2+(|s|^{-1}\xi,\b,\rho)\cdot\psi_2\big)\right)\\
&\lesssim \iint_{\MM} \left(a^\frac 12 |s|^{-1}+\RR |s|^{-1}+|s|\cdot \delta^2 a^\frac 32 |s|^{-4}+|s|\cdot \delta a |s|^{-3}\right) ||s^{i_1+2\lambda-1}\d^{i_1}\psi_2||_{L^2(S)}^2\\
%&\quad +\iint  (a^\frac 12 |s|^{-1}+\RR |s|^{-1}+|s|\cdot \delta^2 a^\frac 32 |s|^{-4}+|s|\cdot \delta a |s|^{-3}) |||s|^{j+2\lambda-1}\psi_2||_{L^2(S)}\\
&\lesssim \delta (a^\frac 12 (a\delta)^{-1})\cdot ||s^{i_1+2\lambda-1}\d^{i_1}\psi_2||_{L^2(\hbub)}^2\lesssim a^{-\frac 12} ||s^{i_1+2\lambda-1}\d^{i_1}\psi_2||_{L^2(\hbub)}^2.
    \end{split}
\end{equation*}
Both terms can be absorbed by the left hand side. This finishes the proof of the estimate \eqref{eq:Main-Bianchipairs}.
\end{proof}

\begin{remark} In the estimates below we apply Lemma    \ref{Lemma:Main-Bianchipairs}   to the actual Bianchi pairs $(\a, \b)$, $\big(\b, (\rho, \rhod)\big)$, $\big((\rho, \rhod), \bb\big)  $ and $(\bb, \aa)$.    In doing this we  will ignore     nonlinear terms of the form $\Gac_g \cdot \psi_1$ in $F _1$ or $ F_2$, and $(\Gac_b,\chih)\cdot \psi_2$ in $ F_2$,  as they have already been dealt with in the proof of  the Lemma.
\end{remark}
\subsection{Estimate of the pair \texorpdfstring{$(\a,\b)$}{(a,b)}}
\begin{proposition}\label{prop:estimate-Bianchi-pair1}
    We have $a^{-\frac 12}||s^i\d^i\a||^2_{L^2(\Sigma_{\tau;\ub})}+||s^i\d^i\b||^2_{L^2(\hbub)}\lesssim \delta^{-1} a$, $i\leq N$.
\end{proposition}
\begin{proof}
    Consider  the     equations 
\beaa
    \nab_3\a-   \nab\hot \b&=&-\frac 1 2 \trchb\a+F _1,\\
\nab_4\beta - \div\a &=& F_2.
\eeaa
with 
\beaa
   F_1&=&
   \zeta\hot \b - 3 (\rho\chih +\rhod\dual\chih),
\\
     F_2&=&-2\trch\beta +2\atrch\dual\b-2\om\b +\a\c  (2\zeta+\etab) + 3  (\xi\rho+\dual \xi\rhod),
\eeaa
 of the form \eqref{eq:null-Bianchi-model-pair}, with $k=2$, $\lambda=\frac 12$, $\psi_1=\a$, $\psi_2=\b$. Therefore, applying Lemma \ref{Lemma:Main-Bianchipairs}, we have
\begin{equation*}
    \begin{split}
         a^{-\frac 12} \sup_\tau\int_{\Sigma_\tau}  s^{2i}|\d^i\a|^2+\sup_{\ub}\int_{\hbub} s^{2i}|\d^i\b|^2&\lesssim \int_{H_{-1}}s^{2i}|\d^i\a|^2+\int_{\Hb_0} s^{2i}|\d^i\b|^2\\
         &+\Big|\iint_{\MM} s^{2i}(\d^i\a\cdot \d^iF_1+\d^i\b\cdot \d^i  F_2)\Big|.
    \end{split}
\end{equation*}
As before, we only need to estimate the nonlinear terms with one of the factors replaced by its $L^\infty(S)$ norm.  Note that at the top order of derivatives, we can only use the weaker estimate of $\xi$ from Proposition \ref{Prop:estimate-xi}.
We have
\begin{equation*}
\begin{split}
    \Big|\iint_{\MM} &s^{2i}\d^i(\Gac_g\cdot\b,\chih\cdot\rho))\cdot\d^i\a+s^{2i}\d^i(\Gac_g \cdot \a)\cdot \d^i\b\Big| \\
    &\lesssim \iint_{\MM} \left(\o\delta a^\frac 12 |s|^{-2} ||s^{i_1}\d^{i_1}\b||_{L^2(S)}+a^\frac 12 |s|^{-1}||s^{i_1}\d^{i_1}\rho||_{L^2(S)}\right)||s^{i_1}\d^{i_1}\a||_{L^2(S)}\\
    &\lesssim \o\delta a^\frac 12 \left(\iint_{\MM} |s|^{-4} |s^{i_1}\d^{i_1}\a|^2\right)^\frac 12 \left(\iint_{\MM} |s^{i_1}\d^{i_1}\b|^2\right)^\frac 12\\
    & \quad +a^\frac 12 \left(\iint_{\MM} |s|^{-4} |s^{i_1}\d^{i_1}\a|^2\right)^\frac 12 \left(\iint_{\MM} |s|^2 |s^{i_1}\d^{i_1}\rho|^2\right)^\frac 12\\
    &\lesssim \o\delta a^\frac 12 (\int_{-1}^{-a\delta} |\tau|^{-4} \sup_\tau a^{-\frac 12}||s^{i_1}\d^{i_1}\a||^2_{L^2(\Sigma_\tau)})^\frac 12 \cdot \delta^\frac 12 \sup_{\ub} ||s^{i_1}\d^{i_1}\b||_{L^2(\hbub)}\\
    &\quad +a^\frac 12 (\int_{-1}^{-a\delta} |\tau|^{-4} \sup_\tau a^{-\frac 12}||s^{i_1}\d^{i_1}\a||^2_{L^2(\Sigma_\tau)})^\frac 12 \cdot \delta^\frac 12 \sup_{\ub}||ss^{i_1}\d^{i_1}\rho||_{L^2(\hbub)} \\
    &\lesssim \o\delta a^\frac 12 |a\delta|^{-\frac 32} \RR\delta^{-\frac 12} a^{\frac 12} \cdot \delta \cdot \RR\delta^{-\frac 12} a^\frac 12+a^\frac 12 |a\delta|^{-\frac 32} \RR\delta^{-\frac 12} a^\frac 12 \cdot \delta^\frac 12  \RR\delta^\frac 12 a^\frac 12\\
    &\lesssim \o\RR^2 \delta^{-1} +\RR^2 \delta^{-1} \lesssim \delta^{-1}a,
\end{split}
\end{equation*}
%\begin{equation*}
%    \d^i  F_2=\d^i((\Gac_b\cdot \b),(\Gac_g\cdot \a),(\xi\cdot \rho)).
%\end{equation*}
and
\begin{equation*}
    \begin{split}
        \Big|\iint_{\MM} s^{2i} \d^i(\xi\cdot\rho)\cdot \d^i\b\Big|
        &\lesssim \iint_{\MM} \o \delta^{-\frac 12} a^\frac 12 |s|^{-\frac 12} ||s^{i_1}\d^{i_1}\rho||_{L^2(S)} ||s^{i_1}\d^{i_1}\b||_{L^2(S)}\\
        &\lesssim  \o \delta^{-\frac 12} a^\frac 12 |a\delta|^{-\frac 32} \delta ||ss^{i_1}\d^{i_1}\rho||_{L^2(\hbub)} ||s^{i_1}\d^{i_1}\b||_{L^2(\hbub)}\\
        &\lesssim \o \delta^{-\frac 12} a^\frac 12 |a\delta|^{-\frac 32}\delta  \RR \delta^\frac 12 a^\frac 12\cdot \RR \delta^{-\frac 12} a^\frac 12\\
        &\lesssim \o\RR^2 \delta^{-1} \lesssim \delta^{-1} a.
    \end{split}
\end{equation*}
\end{proof}

\subsection{Estimate of the pair \texorpdfstring{$(\b,\rho)$}{(b,rho)}}
\begin{proposition}\label{prop:estimate-Bianchi-pair2}
    We have, for all  $i\leq N$.
    \beaa
    a^{-\frac 12} ||s^{i+1}\d^i \b||_{L^2(\Sigma_{\tau;\ub})}+||s^{i+1} \d^i (\rho,\rhod)||^2_{L^2(\hbub)}\lesssim \RR[\a]^2 \delta a.
    \eeaa
\end{proposition}
Note that have shown that $\RR[\a]\lesssim 1$, so this is an improvement of  $C_b^2 \delta a$ in the bootstrap assumption,  if $C_b$  is   sufficiently larger  to start  with.
\begin{proof}
    We start with 
\beaa
\nab_3\b&=&-\DD_1^*(\rho,-\dual\rho)-\trchb \b+F _1
     ,\\
 \nab_4 (\rho, - \dual\rho)&=&\DD_1 \b+ F_2, 
\eeaa
with
\beaa
   F_1&=&2\bb\cdot\chih,
\\
     F_2&=&(2\etab+\zeta)\cdot(\b,\dual\b)-2\xi\cdot (\bb,-\dual\bb)-\frac 12 \chibh\cdot (\a,\dual\a)
\eeaa
of the form \eqref{eq:null-Bianchi-model-pair}, with $k=1$, $\lambda=1$, $\psi_1=\b$, $\psi_2=(\rho,-\dual\rho)$.

We only need to estimate the terms $\xi\cdot\bb$ and $\chibh\cdot\a$ in $ F_2$. We have
\begin{equation*}
    \begin{split}
        \iint_{\MM} |s|^{2+2i} &\d^i(\xi\cdot\bb)\cdot \d^i(\rho,\rhod)\lesssim 
    \iint_{\MM} |s|^{2} \o\RR \delta^{-\frac 12} a^\frac 12 |s|^{-\frac 12}\cdot ||s^{i_1}\d^{i_1}\bb||_{L^2(S)} \cdot ||s^i\d^i\rho||_{L^2(S)}\\
&\lesssim \o\RR\delta \cdot \delta^{-\frac 12} a^\frac 12 |a\delta|^{-\frac 32} \sup_{\ub} ||s^2s^{i_1}\d^{i_1}\bb||_{L^2(\hbub)} ||ss^{i_1}\d^{i_1}\rho||_{L^2(\hbub)}\\
&\lesssim  \o\RR \delta^{-1} a^{-1}\cdot \RR \delta^\frac 32 a\cdot \RR\delta^\frac 12 a^\frac 12=\o\RR^2 \delta a^\frac 12\ll \delta a,
\end{split}
\end{equation*}
and, using the improved estimate of $\chibh$ obtained in Proposition \ref{prop:L2-estimate-wtb,chibh},
\begin{equation*}
    \begin{split}
        \iint_{\MM}& |s|^{2+2i} \d^i(\chibh\cdot\a)\cdot \d^i(\rho,\rhod) \leq  \iint_{\MM} |s|^{-2} \delta a^\frac 12 |s|^{-2} |s^{i_1}\d^{i_1}\a|\cdot|s^i\d^i (\rho,\rhod)|\\
        &\leq \delta a^\frac 12\left(\iint_{\MM} |s|^{-2}|s^{i_1}\d^{i_1}\a|^2\right)^\frac 12 \left(\iint_{\MM} |s|^2|s^{i_1}\d^{i_1}(\rho,\rhod)|^2\right)^\frac 12\\
&\leq  \delta a^\frac 12 \left(\int_{-1}^{-\a\delta} |\tau|^{-2} \left(\sup_\tau a^{-\frac 12}\int_{\Sigma_\tau} |s^{i_1}\d^{i_1}\a|^2\right)\, d\tau%用|s|\geq |\tau|
\right)^\frac 12 \left(\delta\cdot \sup_{\ub}||ss^{i_1}\d^{i_1}\rho||_{L^2(\hbub)}^2\right)^\frac 12\\
&\leq \delta a^\frac 12 |a\delta|^{-\frac 12}\RR[\a]\delta^{-\frac 12} a^\frac 12 \cdot \RR[\rho] \delta a^\frac 12\\
&\leq  \RR[\rho] \RR[\a] \delta a\leq C^{-1} \RR[\rho]^2 \delta a+C\RR[\a]^2\delta a,
    \end{split}
\end{equation*}
so taking a suitable $C>0$, the first term can be absorbed by the left hand side (which is like $\RR[\rho]^2 \delta a$), and we obtain the result.
%\beaa %仍然是最高阶b要用flux，需要注意
%\iint |s|^2 \zeta\cdot\b\cdot \rho &\leq & \iint \delta a^\frac 12 |s|^{-1} |\b|\cdot |s\rho|\\
%&\leq & \RR^2 \delta\cdot \delta a^\frac 12 (a\delta)^{-1} \cdot (\int ||\b||_{L^2(S)}^2 ds)^\frac 12 ||s\rho||_{L^2(\hbub)}\\
%&\leq & \RR^2\delta^2 a^\frac 12 \cdot (a\delta)^{-1} \cdot \delta^{-\frac 12} a^\frac 12 \cdot \delta^\frac 12 a^\frac 12=\RR^2 \delta a^\frac 12.
%\eeaa
%(Note that $\mathcal R[\b]$ is for the pair $(\b,\a)$.)
\end{proof}

\subsection{Estimate of the pair \texorpdfstring{$(\rho,\bb)$}{(rho,bb)}}
\begin{proposition}\label{prop:estimate-Bianchi-pair3}
    We have\footnote{We omit  from the  estimate the flux on $\Sigma_\tau$,  as it is no longer needed. } $||s^{i+2}\d^i \bb||^2_{L^2(\hbub)}\lesssim \delta^3 a^2$, $i\leq N$.
\end{proposition}

\begin{proof}
   Consider the equations
\beaa
\nab_3 (\rho,\rhod)+\DD_1\bb&=&-\frac 3 2 \trchb (\rho,\rhod) +F _1,\\
     \nab_4\bb&=&\DD_1^*(\rho,\rhod).
    \eeaa
    with
    \beaa
       F_1&=&\zeta\cdot (\bb,\dual\bb)-\frac 12\chih\cdot (\ab,\dual\ab),
    \\
     F_2&=&-(\trch \bb+ \atrch \dual \bb)+ 2\om\,\bb+2\b\c \chibh
    -3 (\rho\etab-\rhod\dual \etab)-    \aa\c\xi,
\eeaa
of the type \eqref{eq:null-Bianchi-model-pair2}, with $k=1$, $\lambda=\frac 32$, $\psi_1=(\rho,\rhod)$, $\psi_2=\bb$.
%(targeting for $\delta^3 a^2$).
We have
\begin{equation*}
\begin{split}
    \Big|\iint_{\MM} |s|^{2i+4} &\d^i (\chih\cdot \ab) \cdot \d^i(\rho,\rhod)\Big| \lesssim \iint_{\MM} |s|^4 a^\frac 12 |s|^{-1} |s^{i_1}\d^{i_1}\ab| |s^i\d^i(\rho,\rhod)|\\
    &\lesssim \delta\cdot a^\frac 12 |a\delta|^{-1}\cdot \sup_{\ub}||s^3s^{i_1}\d^{i_1}\ab||_{L^2(\hbub)} ||ss^{i}\d^{i}(\rho,\rhod)||_{L^2(\hbub)}\\
    &\lesssim a^{-\frac 12} \RR[\ab]\delta^\frac 52 a^\frac 32 \cdot \RR[\rho,\rhod] \delta^\frac 12 a^\frac 12\lesssim \RR^2\delta^3 a^\frac 32\ll \delta^3 a^2.
\end{split}
\end{equation*}
Recall  that, see  Proposition \ref{prop:estimate-Bianchi-pair2},
\begin{equation*}
    \sup_\tau a^{-\frac 12}\int_{\Sigma_\tau} | s^{i+1}\d^i\b|^2 \lesssim \RR[\a]^2 \delta a,\quad i\leq N.
\end{equation*}
The $C_b^2$ here can in fact be dropped in view of Proposition \ref{prop:estimate-Bianchi-pair1}.
Then we have
\begin{equation*}
    \begin{split}
        \Big|\iint_{\MM} |s|^{2i+4}& \d^i(\b\cdot\chibh)\cdot \d^i \bb\Big| \lesssim \iint_{\MM} |s|^4 \o \delta a^\frac 12 |s|^{-2} |s^{i_1}\d^{i_1}\b| |s^i\d^i\bb|\\
        &\lesssim \o\delta a^\frac 12  \cdot \left(\iint_{\MM} |s|^{-2}|ss^{i_1}\d^{i_1}\b|^2\right)^\frac 12 \left(\iint_{\MM} |s|^4 |s^i\d^i\bb|^2\right)^\frac 12\\
        &\lesssim \o\delta a^\frac 12 \left(\int_{-1}^{-a\delta} |\tau|^{-2} \RR^2 \delta a\, d\tau\right)^\frac 12  \left(\delta \sup_{\ub}||s^2 s^i\d^i\bb||_{L^2(\hbub)}^2\right)^\frac 12\\
        &\lesssim \o\delta a^\frac 12 \cdot |a\delta|^{-\frac 12} \RR \delta^{\frac 12} a^\frac 12 \cdot \delta^\frac 12 \RR[\bb] \delta^\frac 32 a\\
        &\lesssim \o\RR^2 \delta^3 a^\frac 32\ll \delta^3 a^2.
    \end{split}
\end{equation*}
Also,
\begin{equation*}
\begin{split}
\Big|\iint_{\MM} |s|^{2i+4} &\d^i(\ab\cdot \xi)\cdot\d^i\bb\Big| \lesssim  \iint_{\MM} |s|^4 \o \delta^{-\frac 12} a^\frac 12 |s|^{-\frac 12} |s^{i_1}\d^{i_1}\ab||s^{i}\d^{i}\bb|\\
&\lesssim  \iint_{\MM} \o\delta^{-\frac 12} a^\frac 12 |s|^{-\frac 32} |s^3s^{i_1}\d^{i_1}\aa|\cdot |s^2s^{i}\d^{i}\bb|\\
&\lesssim  \o\delta^{-\frac 12} a^\frac 12 \cdot \delta |a\delta|^{-\frac 32} \sup_{\ub} ||s^3s^{i_1}\d^{i_1}\ab||_{L^2(\hbub)}\cdot ||s^2s^{i}\d^{i}\bb||_{L^2(\hbub)}\\
&\lesssim  \o\delta^{-\frac 12} a^\frac 12 \cdot \delta |a\delta|^{-\frac 32} \cdot \RR \delta^\frac 52 a^\frac 32 \cdot \RR\delta^\frac 32 a\\
&\lesssim \RR^2 \delta^3 a^\frac 32 \ll \delta^3 a^2.
\end{split}
\end{equation*}

\end{proof}

\subsection{Estimate of the pair \texorpdfstring{$(\bb,\aa)$}{(bb,aa)}}
\begin{proposition}
    We have $||s^{i+3}\d^i \ab||_{L^2(\hbub)}\lesssim \delta^5 a^3$, $i\leq N$.
\end{proposition}
\begin{proof}
      Consider the equations
\beaa
\nab_3\bb  &=&-\DD_2\ab-2\trchb\,\bb+F _1 ,\\
     \nab_4\aa&=&2\DD_2^*\bb+ F_2.
\eeaa
with
\begin{equation*}
   F_1=2\aa\c \zeta,
\end{equation*}
\begin{equation*}
     F_2=-\frac 1 2 \trch\aa+\frac 12\atrch\dual\aa+4\om\aa
   +(\zeta-4\etab)\hot \bb - 3  (\rho\chibh -\rhod\dual\chibh).
\end{equation*}
This is of the type \eqref{eq:null-Bianchi-model-pair2}, with $k=2$, $\lambda=2$, $\psi_1=\bb$, $\psi_2=\aa$.

As above, we only need to deal with the term $\rho\chibh$ in $ F_2$ ($\rhod\dual\chibh$ is, of course, similar). We have
\begin{equation*}
    \begin{split}
        \Big|\iint_{\MM} s^{2i+6} \d^i(\chibh\cdot\rho)\cdot \d^i\ab\Big| &\lesssim \iint_{\MM} |s|^{6} \o\delta a^\frac 12 |s|^{-2} |s^{i_1}\d^{i_1}\rho| |s^{i}\d^{i}\ab|\\
        &\lesssim \o\delta a^\frac 12\cdot \delta \cdot \sup_{\ub}||ss^{i_1}\d^{i_1}\rho||_{L^2(\hbub)}||s^3s^{i}\d^{i}\ab||_{L^2(\hbub)}\\
        &\lesssim \o \delta^2 a^\frac 12 \RR \delta^\frac 12 a^\frac 12 \cdot \RR\delta^\frac 52 a^\frac 32 \\
        &\lesssim \o\RR^2 \delta^5 a^\frac 52\ll \delta^5 a^3.
    \end{split}
\end{equation*}
\end{proof}

\subsection{End  of the proof of  Part 1 of    Theorem \ref{theorem:main'} }

According  to Proposition \ref{Prop:summary-Ricci} we have
the following  estimates hold true, see  \eqref{eq:compundnorms} for the definition of the norms,
  \beaa
   \o_{\le N} + \RR_{\le N-1}  \les \RR_{\le N}.
  \eeaa
According to  the results of this  section, we also have
\beaa
 \RR_{\le N}\les 1.
\eeaa
Therefore, combining them together and using the non-integrable Sobolev estimate \eqref{eq:non-integrable-Sobolev-estimate}, we have improved all bootstrap assumptions, and as a result, the spacetime can be extended to $\tau^*=-\frac 18 a\delta$ such that  
the estimates $\OO_{\le N}+\RR_{\le N} \les 1$ remain valid.   This ends the proof of the first part of the  Main Theorem \ref{theorem:main'}.

%%%%%%%%%%%%

\section{Formation of trapped surfaces I}

%%%%%%%%%%%%%%%
%We need to consider $S_{\ub,s}$ flowed from $S_{0,s}$ through $\pa_{\ub}=-(\eg_4(s))\eg_3+\eg_4+b^a\pa_{\theta^a}$. 

We    review   the   original proof of  Christodoulou \cite{Chr1} adapted to our foliation, which relies   on the first part of Theorem   \ref{theorem:main'},  the lower bound  condition 
 \eqref{eq:chih_0-lowerbound} and   a simple ODE argument which we  sketch below.  For simplicity of the presentation, as in \cite{Chr1},    we consider only  the Minkowskian incoming data.

 We first note that the sphere $S_{\delta,-\frac 14 a\delta}$ lies in our constructed region. 
We study the value of the null expansions at each point of this sphere.  Consider  the  vectorfield  $V$  with $V(s)=0$ and $V(\ub)=1$,
\beaa
V=-\big(\eg_4(s)\big)\eg_3+\eg_4. 
\eeaa
In the PT frame we have
\beaa
 V&=&-\big(\eg_4(s)\big) e_3+ e_4+f^a e_a+\frac 14|f|^2 e_3=-\big(e_4(s)+\frac 14|f|^2\big)e_3+e_4+f^a e_a+\frac 14|f|^2 e_3\\
 &=&-\big(e_4(s)\big)e_3+e_4+f^a e_a.
\eeaa
Since, according to  \eqref{eq:eS_4(s)}, $|\eg_4(s)|\lesssim 1$,   using the estimate    $|f|\lesssim \delta a^\frac 12 |s|^{-1}$,  we  infer that  $|e_4(s)|\lesssim 1$.
Therefore, since along the trajectories of $V$ we have $|\trch|,|\atrch|,|\om|\lesssim |s|^{-1}$, and $|\div\xi|\lesssim a^\frac 12 |s|^{-1}$, we  infer that
\begin{equation}\label{eq:transport-v}
\begin{split}
    V(\trch)&=-(e_4(s))\nab_3(\trch)+\nab_4\trch+f^a\nab_a\trch\\
    &=O(|s|^{-2})-|\chih|^2-\frac 12 (\trch)^2+\frac 12 (\atrch)^2-2\om \trch+2\div\xi+2\xi\cdot\zeta+f\cdot \nab\trch\\
    &=O(a^\frac 12 |s|^{-2})-|\chih|^2.
    \end{split}
\end{equation}
where we use the null structure equation $\nab_4\trch=-\frac 12 (\trch)^2-|\chih|^2+\frac 12 (\atrch)^2-2\om\trch+2\div\xi+2\xi\cdot\zeta$.
Note that, by revisiting the estimate in the proof of Proposition \ref{Prop:estimate-of-chih}, one has
\begin{equation*}
    \chih=\chih_0 |s|^{-1}+O(|s|^{-1})
\end{equation*}
where $\chih_0=\chih|_{H_{-1}}$ satisfies $\inf_\theta \int_0^\delta |\chibh_0(\ub,\th)|^2 d\ub\geq \delta a$. In the case of flat incoming data,  we have $\trch=2/|s|$ on $\Hb_0$. Therefore, integrating the equation \eqref{eq:transport-v} along the flow line from some point on $\Hb_0$  to a given point on $S_{\delta,-\frac 14 a\delta}$, we get
\begin{equation*}
\begin{split}
    \trch&=\frac 2{|s|}-\int_0^\delta |\chih|^2 d\ub+O(\delta a^\frac 12 |s|^{-2})\leq \frac 2{|s|}-|s|^{-2}\delta a+O(\delta a^\frac 12 |s|^{-2})\\
    &=|s|^{-2}\big(2|s|-\delta a+O(\delta a^\frac 12)\big),
    \end{split}
\end{equation*}
so taking $s=-\frac 14 \delta a$, we see that $\trch<-\frac 1{64}(\delta a)^{-1}$ at each point on $S_{\delta,-\frac 14 a\delta}$.
Note that this  holds true in the PT frame. Using the transformation formula for $\trch$   in Lemma  \ref{Proposition:transformationRicci}, we deduce,  since  $|f|,|s\d f|,|s\zeta|\lesssim a^{-\frac 12}$,
\begin{equation*}
    \trchg=\trch+O(f\cdot \zeta+|f|^2 \trchb+\d f)=\trch+O\big(a^{-\frac 12}\frac{1}{|s|}\big)\leq -\frac 1{64}(\delta a)^{-1}+O(a^{-\frac 12})(a\delta)^{-1}<0.
\end{equation*}
This means that the outgoing null expansion in the (integrable) PG  frame is negative. The ingoing null expansion is also negative in view of the ingoing Raychauhuri equation (i.e., the equation of $\nab_3\trch$). Therefore, we see that $S_{\delta,-\frac 14 a\delta}$ is a trapped surface. 
\begin{remark}
Note that by the Raychaudhuri equation $\nab_4\trch=-\frac 12(\trch)^2-|\chibh|^2$ on $H_{-1}$ and the bounds $|\trch|\lesssim 1$, $|\chih|\lesssim a^\frac 12$ there, we have $\trch=1+O(\delta a)>0$ on $H_{-1}$, so that we can rule out   trapped surfaces  on the initial outgoing null hypersurface.
\end{remark}

\def\const{\mathrm{const}}

\section{Formation of Trapped Surfaces II}

We show here how to derive more precise  estimates for   our main  non-trivial, large,    quantities    $ \chih, \trch, \rho$. As a result, one can also prove the formation of trapped surfaces (the second part of Theorem \ref{theorem:main'}) without using $e_4$ transport equation of Ricci coefficients beyond $H_{-1}$.
 To illustrate the idea it is  easier to  work in  the integrable  PG  frame in this section, without  reference to the PT frame.
 For simplicity,  since there is no  possible confusion, we will  thus  suppress the prefix $^{(g)}$.  We also again assume the data on $\Hb_0$ is the Minkowski data.

For the initial data on $H_{-1}$ and the geodesic coordinate $(\ub,\theta)$, we write
\begin{equation*}
    \chih_{ab}|_{H_{-1}}=a^\frac 12 I(\ub,\theta)_{ab}.
\end{equation*}
So roughly one can think of $|I(\ub,\theta)|\sim 1$.
By our assumption \eqref{eq:chih_0-lowerbound} we have
\begin{equation}\label{eq:I-lower-bound}
    \inf_\theta \int_0^{\delta} |I(\ub,\theta)|^2 d\ub\geq \delta.
\end{equation}
We will define the derivative and integral of $I$ in $\ub$ later in this section. Denote them by $\dot I$ and $\int I$. 
We will prove the following estimates in $\MM$:
 \begin{equation*}
     \chih=a^\frac 12 I(\ub,\theta)|s|^{-1}+O(|s|^{-1}),\quad  \a=a^\frac 12 \dot I(\ub,\theta) |s|^{-1}+O(\delta^{-1}|s|^{-1}),
 \end{equation*}
 \begin{equation*}
     \chibh=a^\frac 12 \left(\int_{0}^{\ub} I(\ub',\theta)\, d\ub'\right) |s|^{-2}+O(\delta |s|^{-2}),
 \end{equation*}
 \begin{equation*}
     \trch=2|s|^{-1}-a\left(\int_0^{\ub} |I(\ub',\theta)|^2\, d\ub'\right)|s|^{-2}+O(\delta |s|^{-1}+\delta a^\frac 12 |s|^{-2}),
 \end{equation*}
 \begin{equation*}
     \rho=\frac 12 a \left(\int_0^{\ub} \dot I(\tilde\ub,\theta)_{ab} \left(\int_0^{\tilde\ub} I(\ub',\theta)_{ab}\, d\ub'\right) d\tilde\ub\right) |s|^{-3}+O(\delta a^\frac 12 |s|^{-3}).
 \end{equation*}
The size of the remainders are smaller than the original bounds of the corresponding quantities at least by a factor of $a^\frac 12$. In particular, the expansion of $\trch$ here will immediately imply the formation of trapped surfaces: The estimate gives
\begin{equation*}
    \trch \leq (2+\delta)|s|^{-1}-a \left(\int_0^{\ub} |I(\ub',\theta)|^2 \, d\ub'\right) |s|^{-2}.
\end{equation*}
So by the lower bound assumption, we have
\begin{equation*}
    \trch\leq (2+\delta)|s|^{-1}-a\delta |s|^{-2}=|s|^{-1}(2+\delta-a\delta |s|^{-1}),
\end{equation*}
so taking $s=-\frac 14 a\delta$ we get trapped surfaces. 
\begin{remark}\label{rmk:aniso-Ch8}
    This precise upper bound of $\trch$, which only depends on the initial data, is in fact everything needed to prove the formation of trapped surfaces in anisotropic situations, i.e. with the inf in \eqref{eq:I-lower-bound} replaced by sup, in view of the argument in \cite{KLR} (see also \cite{AnHan}).
\end{remark}

\subsection{Initial data on \texorpdfstring{$H_{-1}$}{H(-1)}}
On $H_{-1}$, recall that we chose $e_4$ to be the geodesic vector field (equal to the given choice on $S_{0,-1}$). This determines an affine parameter on $s$ on $H_{-1}$ satisfying $e_4(\ub)=1$. Let $\{e_a,e_b\}$ be an orthonormal basis of the spheres given by $\ub=\const$. We write
\beaa
    \chih\big|_{H_{-1}}(\ub,\theta)=     a^{1/2}     I(\ub,\theta)\,
\eeaa 
where $I \in\sk_2 $ is an horizontal,  symmetric, traceless $2$-tensor.

The indices $a,b$ are of course covariant under the orthogonal transformation of $\{e_a,e_b\}$. In order to integrate   transport equations of tensors, we choose the Fermi frame {on $H_{-1}$}  by requiring  
\begin{equation*}
    D_4 e_a=-\zeta_a e_4,
\end{equation*}
{and $e_a$ is  given as an arbitrary   horizontal    frame  on $\Hb_0\cup H_{-1}$.  }
%Recall $\zeta=\frac 12 g(D_a e_4,e_3)$ is well-defined on $H_{-1}$. 
{In particular, the frame satisfies $\nab_4 e_a=0$ on $H_{-1}$ and as a result, for any horizontal covariant tensor $\psi_{a_1\cdots a_k}$, we have $\nab_4 \psi_{a_1\cdots a_k}=e_4(\psi_{a_1\cdots a_k})$ under this frame. This allows us to define the integration and differentiation of $\psi$ through the scalar fields obtained by the components of $\psi$ under this frame.}

One can verify that this frame is indeed tangent to $\{\ub=\const\}$ {on $H_{-1}$: Since $e_4(\ub)=1$, we have $e_4(e_a(\ub))=[e_4,e_a](\ub)=(\etab+\zeta)_a e_4(\ub)-\chi_{ab} e_b(\ub)=\chi_{ab} e_b(\ub)$ using that $\etab+\zeta=0$, which comes from the fact that $e_4$ is geodesic. Hence $e_a(\ub)=0$ on $H_{-1}$. It is also straightforward to verify that $\{e_a\}$ remains an orthonormal frame.}

Recall that we have the bounds on $H_{-1}$ from Proposition \ref{Prop:InitailData-H_{-1}}
\begin{equation*}
    |\chih|\lesssim 1,\quad |\trchb+2|,|\chibh|, |\nab^{\leq 1}\eta|\lesssim \delta a^\frac 12, \quad |\trch|\lesssim 1,\quad |\b|\lesssim a^\frac 12.
\end{equation*}
We now derive more precise estimates for various quantities on $H_{-1}$.
Integrating the equation
\begin{equation*}
    \nab_4 \chibh_{ab}=-\frac 12 \trchb \chih_{ab}-\frac 12 \trch\chibh_{ab}+(\nab\hot\etab)_{ab}+(\etab\hot\etab)_{ab}
\end{equation*}
in the Fermi frame, we obtain,
\begin{equation*}
    \chibh_{ab}=\int_0^{\ub} a^\frac 12 I(\ub',\theta)_{ab}\, d\ub'+O(\delta^2 a^\frac 12)
\end{equation*}
where the integral is defined componentwise under the Fermi frame.

For $\trch$ we have, on $H_{-1}$,
\begin{equation*}
    \nab_4\trch =-|\chih|^2-\frac 12 (\trch)^2.
\end{equation*}
Since $\trch=2$ on $H_{-1}\cap \Hb_0$, combining it with the weak bound $|\trch|\lesssim 1$ we have
\begin{equation*}
    \trch= 2-a\int_0^{\ub} |I(\ub',\theta)|^2\, d\ub'+O(\delta).
\end{equation*}
For $\a$, since
\begin{equation*}
    \nab_4 \chih=-\trch\chih-\a,
\end{equation*}
we have
\begin{equation*}
    \a_{ab}=-e_4 (\chih_{ab})-\trch\chih_{ab}=-\dot I(\ub,\theta)_{ab} a^\frac 12+O(a^\frac 12),
\end{equation*}
where $\dot I(\ub,\theta)_{ab}:=\nab_4 I(\ub,\theta)_{ab}$. It is also equal to the $\ub$-derivative of $I(\ub,\theta)_{ab}$ as scalars under the Fermi frame. 

Then we turn to the equation of $\rho$
\begin{equation*}
    \nab_4\rho=\div\b-\frac 32 \trch\rho+2(\etab+\zeta)\cdot\b-\frac 12 \chibh\cdot\a.
\end{equation*}
by the rough bounds we see that the main contribution comes from the last term. This gives, on $H_{-1}$,
\begin{equation*}
    \rho=\frac 12 a \left(\int_0^{\ub} \dot I(\tilde\ub,\theta)_{ab} \left(\int_0^{\tilde\ub} I(\ub',\theta)_{ab}\, d\ub'\right) d\tilde\ub\right)+O(\delta a^\frac 12).
\end{equation*}
This completes the setup on $H_{-1}$. 

\subsection{Estimate along the \texorpdfstring{$e_3$}{e3} directions}
In the propagation along $e_3$ direction, we also choose the Fermi frame, i.e., 
\begin{equation*}
    D_3 e_a=-\zeta_a e_3
\end{equation*}
with $e_a$ coinciding the one we chose on $H_{-1}$. Since $s=-1$ and hence $e_a(s)=0$ on $H_{-1}$, such a frame will be tangent to the $s$-sections on $\hbub$, so in other words, tangent to $S_{\ub,s}$. Therefore it corresponds to the integrable geodesic frame.
Then similarly, the projected differentiation $\nab_3$ is the same as the $e_3$ derivative of the components under this frame.

To propagate into the interior, we design the linearization based on the estimates we have on $H_{-1}$, and the $|s|$-weight in each $\nab_3$ transport equation. This gives the following ans\"{a}tze:
\def\chihc{\widecheck{\chih}}
\def\chibhc{\widecheck{\chibh}}
\def\trchc{\widecheck{\trch}}
\def\rhoc{\widecheck{\rho}}
\def\ac{\widecheck{\a}}
\def\trchbc{\wtb}
\begin{equation*}
    \chihc_{ab}:=\chih_{ab}-a^\frac 12 I(\ub,\theta)_{ab} |s|^{-1},\quad \chibhc_{ab}:=\chibh_{ab}-a^\frac 12 \left(\int_0^{\ub} I(\ub',\theta)_{ab}\, d\ub'\right)|s|^{-2},
\end{equation*}
\begin{equation*}
    \trchbc:=\trchb+2|s|^{-1},\quad \trchc:=\trch-\left(2|s|^{-1}-a\Big(\int_0^{\ub} |I(\ub',\theta)|^2\, d\ub'\Big)|s|^{-2}\right),
\end{equation*}
\begin{equation*}
    \ac_{ab}:=\a_{ab}+ a^\frac 12 \dot I(\ub,\theta)_{ab} |s|^{-1}, \quad \rhoc=\rho-\frac 12 a \left(\int_0^{\ub} \dot I(\tilde\ub,\theta)_{ab} \Big(\int_0^{\tilde\ub} I(\ub',\theta)_{ab}\, d\ub'\Big) d\tilde\ub\right) |s|^{-3}.
\end{equation*}
\begin{lemma}
    We have
    \begin{equation*}
        \int_0^{\ub} \dot I(\tilde\ub,\theta)_{ab} \left(\int_0^{\tilde\ub} I(\ub',\theta)_{ab}\, d\ub'\right) d\tilde\ub=I(\ub,\theta)_{ab}\int_0^{\ub} I(\ub,\theta)_{ab}\, d\ub'-\int_0^{\ub} |I(\ub',\theta)|^2 d\ub'.
    \end{equation*}
\end{lemma}
\begin{proof}
    They are both zero when $\ub=0$, so it suffices to show that their derivatives in $\ub$ are the same, which is obvious.
\end{proof}
We also remark that this expression, while defined under a special choice of the horizontal frame, does not depend on this choice: The tensor $\int_0^{\ub} I(\ub',\th)\, d\ub'$ on $H_{-1}$ is the unique tensor $K_{ab}$ satisfying $\nab_4 K_{ab}=I_{ab}$, $K_{ab}|_{\ub=0}=0$. The expression $I\cdot K-\int_0^{\ub}|I|^2 d\ub'$ is then clearly a scalar field independent of the choice of $\{e_a\}$.

\def\trchl{\trchc}
\def\trchbl{\trchbc}

We now derive the $e_3$-transport equations of these ``linearized" quantities. First note that, by the estimates above, we have on $H_{-1}$
\begin{equation*}
    \chihc=0,\quad \chibhc=O(\delta^2 a^\frac 12)\ll \delta,\quad \trchl=O(\delta),\quad \ac=O(a^\frac 12)\ll \delta^{-1},\quad \rhoc=O(\delta a^\frac 12).
\end{equation*}
\begin{proposition}
    We have
    \begin{equation*}
        \begin{split}
            \nab_3 \trchbl+\trchb\trchbl&=\frac 12(\trchbl)^2-|\chibh|^2,\\
            \nab_3 \chihc+\frac 12 \trchb\, \chihc&=-\frac 12 a^\frac 12 I |s|^{-1} \trchbl-\frac 12 \trch\chibh+\nab\hot\eta+\eta\hot\eta\\
            \nab_3\chibhc+\trchb\, \chibhc&=-a^\frac 12 \Big(\int I\Big)|s|^{-2}\trchbl-\ab,\\
            \nab_3\trchl+\frac 12\trchb\trchl&=-a^\frac 12 \Big(\int I\Big)|s|^{-2}\cdot \chihc-a^\frac 12 I |s|^{-1} \cdot\chibhc-\chihc\cdot\chibhc+2\div\eta+2|\eta|^2\\
    &\quad -\frac 12\left(2|s|^{-1}-a\Big(\int |I|^2\Big)|s|^{-2}\right) \trchbl+2\rhoc ,\\
            \nab_3 \ac+\frac 12 \trchb\ac&=\nab\hot\b+\frac 12 \trchbc \, a^\frac 12 \dot I |s|^{-1}+\zeta\hot\b-3(\rho\chih+\dual\rho\dual\chih),\\
            \nab_3\rhoc-3|s|^{-1} \rhoc&=-\div\bb
            -\frac 32 \trchbc \rho+\zeta\cdot\bb-\frac 12 \chih\cdot\ab.
        \end{split}
    \end{equation*}
\end{proposition}
\begin{proof}
    Direct calculation of the components in the Fermi frame. Note the presence of $\eta$ terms in the integrable frame (they are nevertheless lower order terms). We only present the longest calculation for $\trchl$. We have (for simplicity, denote $\int=\int_{0}^{\ub}$)
\begin{equation*}
\begin{split}
    &\nab_3 \trchl=\nab_3 \trch-2|s|^{-2}+2a\Big(\int |I|^2 \Big)|s|^{-3}\\
    &=-\chibh\c\chih -\frac 1 2 \trchb\trch+2\rho+2\div\eta+2|\eta|^2-2|s|^{-2}+2a\Big(\int |I|^2 \Big)|s|^{-3}\\
    &=-a I\cdot \Big(\int I\Big)|s|^{-3}-a^\frac 12 \Big(\int I\Big)|s|^{-2}\cdot \chihc-a^\frac 12 I |s|^{-1} \cdot\chibhc-\chihc\cdot\chibhc\\
    &\quad +2|s|^{-2}-a\Big(\int |I|^2\Big)|s|^{-3}-\frac 12\left(2|s|^{-1}-a\Big(\int |I|^2\Big)|s|^{-2}\right) \trchbl+|s|^{-1}\trchl\\
    &\quad -\frac 12 \trchbl\trchl+a\left(\int \dot I \cdot \Big(\int I\Big)\right)|s|^{-3}+2\rhoc-2|s|^{-2}+2a \Big(\int |I|^2\Big)|s|^{-3}+2\div\eta+2|\eta|^2\\
    &=-a^\frac 12 \Big(\int I\Big)|s|^{-2}\cdot \chihc-a^\frac 12 I |s|^{-1} \cdot\chibhc-\chihc\cdot\chibhc+2\div\eta+2|\eta|^2\\
    &\quad -\frac 12\left(2|s|^{-1}-a\Big(\int |I|^2\Big)|s|^{-2}\right) \trchbl+|s|^{-1}\trchl-\frac 12 \trchbl\trchl+2\rhoc,
\end{split}   
\end{equation*}
where we used $\int \dot I \cdot (\int I)=I\int I-\int |I|^2$ shown above.
\end{proof}
Recall that the estimate established in previous sections easily imply the following bounds in the integrable geodesic frame:
\begin{equation*}
    |(s\nab)^{\leq 1}\eta|+|\zeta|+|\trch-\frac{2}{|s|}|+|\chih|\lesssim \delta a^\frac 12 |s|^{-2}, \quad a^{-\frac 12}|\chih|+|\trch|\lesssim |s|^{-1},
\end{equation*}
\begin{equation*}
    |\a|\lesssim \delta^{-1} a^\frac 12 |s|^{-1},\quad |\b|\lesssim a^\frac 12 |s|^{-2},\quad |\rho|\lesssim \delta a|s|^{-3},\quad |\bb|\lesssim \delta^2 a^\frac 32 |s|^{-4},\quad |\ab|\lesssim \delta^\frac 52 a^\frac 32 |s|^{-\frac 92}.
\end{equation*}
This implies, using that $|I|\sim 1$, and $|s|\gtrsim a\delta$,
    \begin{equation*}
        \begin{split}
        \nab_3 \trchbl-\frac{2}{|s|}\trchbl&=O(\delta^2 a |s|^{-4}),\\
            \nab_3 \chihc-|s|^{-1} \chihc&=O(\delta a |s|^{-3})\\
            \nab_3\chibhc-\frac{2}{|s|}\chibhc&=O(\delta^2 a |s|^{-4}),\\
            \nab_3 \ac-|s|^{-1}\ac&=O(\delta a^\frac 32 |s|^{-4}),\\
            \nab_3\rhoc-3|s|^{-1} \rhoc&=
            O(\delta^2 a^\frac 32 |s|^{-5}).
        \end{split}
    \end{equation*}
This immediately implies
\begin{equation*}
    |\trchbl|\lesssim \delta |s|^{-2},\quad |\chihc|\lesssim \frac{1}{|s|},\quad |\chibhc|\lesssim \delta |s|^{-1},\quad |\ac|\lesssim \delta^{-1} |s|^{-1},\quad |\rhoc|\lesssim \delta a^\frac 12 |s|^{-3}.
\end{equation*}
Finally, we derive the estimate of $\trchl$. The estimates above imply
\begin{equation*}
    \nab_3\trchl-|s|^{-1}\trchl=O(\delta a^\frac 12 |s|^{-3}),
\end{equation*}
so we obtain
\begin{equation*}
    |\trchl|\lesssim \delta |s|^{-1}+\delta a^\frac 12 |s|^{-2}.
\end{equation*}
We see that all linearized quantities behave better than the original ones at least by an $a^\frac 12$ factor.

\appendix
\section{Proof of the Sobolev embeddings}
The appendix is the analogue of the proof in \cite{Chr1} with the $e_4$ direction replaced by the $e_3$ direction. Recall the Sobolev estimate on a $2$-Riemannian manifold $S$ for any tensor $\psi$ (See Lemmas 5.1 and 5.2 in \cite{Chr1})
\begin{equation*}
    ||\psi||_{L^\infty(S)}\leq C \max\{I(S),1\} \big(A(S)^\frac 12||(\nab^{S})^2\psi||_{L^2(S)}+||\nab^S \psi||_{L^2(S)}+A(S)^{-\frac 12}||\psi||_{L^2(S)}\big),
\end{equation*}
where $A(\cdot)$ denotes the area, and $I(S)$ is the isoperimetric constant
\begin{equation*}
    I(S):=\sup_{\substack{U\subset S\\ \pa U\in C^1}}\frac{\min\{A(U), A(U^c)\}}{\mathrm{Perimeter}(\pa U)}.
\end{equation*}
To prove the Sobolev embedding, one needs a bound on the isoperimetric constant $I(S)$. Recall that the flow of $\pa_s=e_3$ induces a diffeomorphism between $S_{\ub,s}$ and $S_{\ub,-1}$, and the pullback metric satisfies $\pa_s \gamma=2\chib$. Then for
\begin{equation*}
    \kappa(s)=\frac{d\vol(s)}{d\vol(-1)}
\end{equation*}
one has $\pa_s \log\kappa=\trchb=-\frac{2}{|s|}+\wtb$, which, by the bootstrap bounds, gives
\begin{equation*}
  |\pa_s\log \kappa(s)+\frac{2}{|s|}|\leq ||\wtb||_{L^\infty(S_{\ub,s})}\leq \mathcal O\delta a^\frac 12|s|^{-2},
\end{equation*}
so
\begin{equation*}
    |\log (\frac{\kappa(s)}{|s|^{2}})|=\mathcal O\delta a^\frac 12 |s|^{-1}\ll 1,
\end{equation*}
and hence we get
\begin{equation*}
    C_1 |s|^2\leq \kappa(s)\leq C_2 |s|^2
\end{equation*}
for each point on the sphere $S_{\ub,-1}$.
This gives a control of the area of a region by its pre-image on $S_{\ub,-1}$.

To control the arc length,
consider the matrix of $\gamma(s)$ ($=\slashed g|_{S_{\ub,s}}$) with respect to $\gamma(-1)$. Denote $\Lambda(s)$ and $\lambda(s)$ as the larger and smaller eigenvalue of the matrix. We want a control of $\Lambda(s)$ and $\lambda(s)$. We have studied the quantity
\begin{equation*}
    \kappa(s):=\frac{d\vol(s)}{d\vol(-1)}=\sqrt{\Lambda(s) \lambda(s)}
\end{equation*}
since the determinant is the product of eigenvalues.
Now we define $\hat\gamma(s)=(\kappa(s))^{-1}\gamma(s)$, which is expected to have an almost constant scaling, and we can verify that
\begin{equation*}
    \pa_s\hat\gamma(s)=\frac{\pa_s\gamma(s)}{\kappa(s)}-\frac{\gamma(s)}{(\kappa(s))^2}\pa_s\kappa(s)=\frac{1}{\kappa(s)}2\chib-\frac{\gamma(s)}{\kappa(s)}\trchb=\frac{2}{\kappa(s)}\chibh.
\end{equation*}
Now define
\begin{equation*}
    \nu(s):=\sup_{|X|_{\gamma(-1)}=1} \hat\gamma(s)(X,X).
\end{equation*}
Now for $|X|_{\gamma(-1)}=1$, we have
\begin{equation*}
    \hat\gamma(s)(X,X)=\hat\gamma(-1)(X,X)+2\int_{-1}^s \frac{\chibh(s')(X,X)}{\kappa(s')} ds'.
\end{equation*}
%Now we can expand $\chibh(s')(X,X)$ (a scalar) under the metric $\gamma(-1)$. 
Since $\chibh(s')$ is a $2$-covector, we have
\begin{equation*}
    \begin{split}
        |\chibh(s')|^2_{\gamma(-1)}&=\chibh(s')_{ab}\chibh(s')_{cd}\gamma(-1)^{ac}\gamma(-1)^{bd}\leq \Lambda(s')^2 \chibh(s')_{ab}\chibh(s')_{cd}\gamma(s')^{ac}\gamma(s')^{bd}\\
        &=\Lambda(s')^2|\chibh(s')|_{\gamma(s')},
    \end{split}
\end{equation*}
so
\begin{equation*}
    \hat\gamma(s)(X,X)=\hat\gamma(-1)(X,X)+2\int_{-1}^s \frac{\Lambda(s')|\chibh(s')|}{\kappa(s')} ds'.
\end{equation*}
By definition we get $\nu(s')=\Lambda(s')/\kappa(s')$, so we get
\begin{equation*}
    \nu(s)\leq 1+2\int_{-1}^s \nu(s') |\chibh(s')| ds'.
\end{equation*}
By Gr\"{o}nwall's inequality, we obtain
\begin{equation*}
    \nu(s')\lesssim \exp\Big(\int_{-1}^s |\chibh(s')|ds'\Big)\lesssim 1.
\end{equation*}
Therefore, $\Lambda(s)/\kappa(s)\lesssim 1$ (and hence $\Lambda(s)/\lambda(s)\lesssim 1$). This means that $\lambda(s)\geq C|s|$, so $I(S_{\ub,s})\lesssim I(S_{\ub,-1})$. The comparison of $S_{\ub,-1}$ and $S_{0,-1}$ is similar (using the flow of $e_4$ on $H_{-1}$), and has been established in \cite{Chr1}. Hence we obtain the estimate
\begin{equation*}
    ||\psi||_{L^\infty(S_{\ub,s})}\lesssim ||s^{i-1}(\nab^S)^i\psi||_{L^2(S_{\ub,s})}.
\end{equation*}

\medskip

\vspace{2ex}

\noindent{\bf Acknowledgements.} Both authors would like to thank IHES for their hospitality during visits. 
X.C. thanks his advisor, Hans Lindblad, for supporting his many visits to Columbia University 
through the Simons Collaboration Grant 638955. Klainerman  was supported  by the NSF grant 2201031.

\end{document}